\documentclass[sigconf]{acmart}

\usepackage{booktabs} 
\usepackage{balance} 

\usepackage{amssymb,amsmath}
\usepackage{xcolor}
\usepackage{listings}
\usepackage{float}
\usepackage{tikz}
\usepackage{enumitem}
\usepackage{showexpl}
\usetikzlibrary{automata,shapes.multipart}
\usetikzlibrary[arrows,decorations.pathmorphing,backgrounds,positioning,fit,petri]
\usepackage{graphics}
\usepackage{rotating}
\usepackage{xspace}
\usepackage{colortbl}
\usepackage{tabularx}
\usepackage{multirow}
\usepackage{url}
\usepackage{booktabs}
\usepackage{color}
\usepackage{subfig}
\usepackage{algorithm}
\usepackage{algpseudocode}
\usepackage{hhline}

\newcommand{\bn}{{\sf BN}}
\newcommand{\naturals}{\mathbb{N}}
\newcommand{\graph}{\mathcal{G}}
\renewcommand{\state}{{\bf s}}
\newcommand{\tstate}{{\bf t}}
\newcommand{\update}{{\xi}}
\newcommand{\St}{{\bf S}}
\newcommand{\T}{{\bf T}}
\newcommand{\x}{{\bf x}}
\newcommand{\f}{{\bf f}}

\newcommand{\ts}{{\sf TS}}
\newcommand{\tshat}{\overline{\sf TS}}
\newcommand{\reach}{{\sf reach}}
\newcommand{\scc}{{\sf SCC}}

\newcommand{\prt}{{\sf par}}
\newcommand{\blocks}{\mathcal{B}}
\newcommand{\B}{\overline{B}}
\newcommand{\cross}{\otimes}

\newcommand{\ctr}{{\sf ctr}}
\newcommand{\control}{{\sf C}}
\newcommand{\cset}{{\mathbb{C}}}

\newcommand{\ac}{{\sf ac}}
\newcommand{\bas}{{\sf bas}}
\newcommand{\SB}{{\sf SB}}

\newcommand{\wb}{{\sf WB}}
\newcommand{\hd}{{\sf hd}}

\newcommand{\move}[1]{\stackrel{#1}{\longrightarrow}}
\renewcommand{\path}{\rho}
\newcommand{\pre}{{\sf pre}}
\newcommand{\post}{{\sf post}}

\newtheorem{observation}{Observation}

\newcolumntype{C}[1]{>{\centering\arraybackslash}p{#1}}

\setcopyright{none}
\acmPrice{}
\acmISBN{}

\begin{document}
\title[Decompositional Boolean Network Control]{A Decomposition-based Approach towards the Control of Boolean Networks (Technical Report)}

\author{Soumya Paul}
\affiliation{
\institution{SnT, University of Luxembourg}
}
%
\author{Cui Su}
\affiliation{
\institution{SnT, University of Luxembourg}
}
%
\author{Jun Pang}
\affiliation{
\institution{SnT \& FSTC, University of Luxembourg}
}
%
\author{Andrzej Mizera}
\affiliation{
\institution{LCSB, University of Luxembourg\\
DII, Luxembourg Institute of Health}
}

\renewcommand{\shortauthors}{S. Paul et al.}

\begin{abstract}
We study the problem of computing a minimal subset of nodes of a given asynchronous Boolean network
that need to be controlled to drive its dynamics from an initial steady state (or {\em attractor}) to a target steady state.
Due to the phenomenon of state-space explosion, a simple global approach that performs computations on the entire network,
may not scale well for large networks.
We believe that efficient algorithms for such networks must exploit the structure of the networks together with their dynamics.
Taking such an approach, we derive a decomposition-based solution to the minimal control problem
which can be significantly faster than the existing approaches on large networks.
We apply our solution to both real-life biological networks and randomly generated networks, demonstrating promising results.
\end{abstract}


\maketitle

\section{Introduction}
\label{sec:intro}

Cell reprogramming is a way to change one cell phenotype to another, allowing tissue or neuron regeneration techniques.
Recent studies have shown that differentiated adult cells can be reprogrammed to embryonic-like pluripotent state or
directly to other types of adult cells without the need of intermediate reversion to pluripotent state~\cite{Nature09,SC14}.
This has led to a surge in regenerative medicine and there is a growing need for the discovery of new and efficient methods
for the control of cellular behaviour.

In this work we focus on the study and control of gene regulatory networks (GRNs) and their combined dynamics with an associated
signalling pathway. GRNs are graphical diagrams visualising the relationships between genes and their regulators. They represent
biological systems characterised by the orchestrated interplay of complex interactions resulting in highly nested feedback and
feed-forward loops. Signalling networks consist of interacting signalling pathways that perceive the changes in the environment
and allow the cell to correctly respond to them by appropriately adjusting its gene-expression. These pathways are often complex,
multi-component biological systems that are regulated by various feedbacks and that interfere with each other via diverse
cross-talks. As a result, GRNs with integrated signalling networks are representatives of complex systems characterised by
non-linear dynamics. These factors render the design of external control strategies for these biological systems a very
challenging task. So far, no general mathematical frameworks for the control of this type of systems have been developed~\cite{LSB11,GLDB14,Lai14}.

Boolean networks (BNs), first introduced by Kauffman~\cite{KS69}, is a popular and well-established framework for modelling GRNs
and their associated signalling pathways. Its main advantage is that it is simple and yet able to capture the important dynamic
properties of the system under study, thus facilitating the modelling of large biological systems as a whole. The states of a BN
are tuples of 0s and 1s where each element of the tuple represents the level of activity of a particular protein in the GRN or
the signalling pathway it models - 0 for inactive and 1 for active. The BN is assumed to evolve dynamically by moving from one
state to the next governed by a Boolean function for each of its components. The steady state behaviour of a BN is given by its
subset of states called {\em attractors} to one of which the dynamics eventually settles down. In biological context, attractors
are hypothesised to characterise cellular phenotypes~\cite{KS69} and also correspond to functional cellular states such as
proliferation, apoptosis differentiation etc.~\cite{HS01}.

Cellular reprogramming, or the control of the GRNs and their signalling pathways therefore amount to being able to drive the
dynamics of the associated BN from an attractor to another `desirable' target attractor by controlling or reprogramming the
nodes of the BN. This needs to be done while respecting certain constraints viz. a minimal subset of nodes of the BN are
controlled or the control is applied only for a minimal number of time steps. Under such constraints, it is known that the problem
of driving the BN from a source to a target attractor (the control problem) is computationally difficult~\cite{MHP16,MHP17} and
does not scale well to large networks. Thus a simple global approach (see Section~\ref{ssec:controlProblem} for a description)
treating the entire network in one-go is usually highly inefficient. This is intuitively due to the infamous state-space explosion
problem. Since most practical real-life networks are large, there is a strong need for designing algorithms which exploit
certain properties (structural or dynamic or both) of a BN and is able to efficiently address the control problem.

\smallskip
\noindent{\bf Our contributions.} In this paper, we develop a generic approach towards solving the minimal control problem (defined formally in Section \ref{sec:pre}) on large BNs based on combining both their structural and the dynamic properties. We show that:
\begin{itemize}
\item The problem of computing the minimal set of nodes to be controlled in a single time-step (simultaneously) to drive the system from a source state $\state$ to a target attractor $A_t$ ({\em driver nodes}) is equivalent to computing a subset of states of the state transition graph of the BN called the {\em strong basin} (defined in Section \ref{sec:pre}) of attraction of $A_t$ (dynamic property).
\item  We show how the network structure of a large BN can be explored to decompose it into smaller {\em blocks}. The strong basins of attractions of the projection of $A_t$  to these blocks can be computed {\em locally} and then combined to recover the {\em global} strong basin of attraction of $A_t$ (structural property).
\item Any algorithm for the computation of the global strong basin of attraction of $A_t$ can also be used to compute the local strong basins of attraction of the projections of $A_t$ to the blocks of BN. Doing so results in the improvement in efficiency for certain networks which have modular structures (like most real-life biological networks).
\item We concretise our approach by describing in detail one such algorithm (Algorithm \ref{alg:stb}) which is based on the computation of fixed points of set operations.
\item We have implemented our decomposition-based approach using this algorithm and applied it to a number of case studies of BNs corresponding to real-life biological networks and randomly generated BNs. Our results show that for certain structurally well-behaved BNs our decomposition-based approach is efficient and outperforms the global approach.
\end{itemize}


\section{Related Work}
\label{sec:related}

In recent years, several approaches have been developed for the control of
complex networks~\cite{LSB11,ABGD13,BAGD13,GLDB14,ZA15,CGCK16,L16,MHP16,MHP17,ZYA17}.
Among them, the methods~\cite{LSB11,GLDB14,CGCK16} were proposed to tackle the control of networks with linear time-invariant dynamics. 
Liu et al.~\cite{LSB11} first developed a structural controllability framework for complex networks to solve full control problems,
by identifying the minimal set of (driver) nodes that can steer the entire dynamics of the system. 
Afterwards, Gao et al.\ extended this method to the target control of complex networks~\cite{GLDB14}.
They proposed a $k$-walk method and a greedy algorithm to identify a set of driver nodes
for controlling a pre-selected set of target nodes.
However, Czeizler et al.~\cite{CGCK16} proved that it is NP-hard to find the minimal set of driver nodes for structural target control problems
and they improved the greedy algorithm~\cite{GLDB14} using several heuristics. 
The above methods have a common distinctive advantage that they are solely based on the network structures, 
which are exponentially smaller than the number of states in their dynamics. Nevertheless, they are only applicable to systems with linear time-invariant dynamics. 

The control methods proposed in~\cite{ABGD13,BAGD13,ZA15,L16,MHP16,MHP17,ZYA17} are designed for networks governed by non-linear dynamics. 
Among these methods, the ones based on the computation of the feedback vertex set (FVS)~\cite{ABGD13,BAGD13,ZYA17}
and the `stable motifs' of the network~\cite{ZA15} drive the network towards a target state by regulating a component of the network with some constraints (feedback vertex sets and stable motifs).
The method based on FVS is purely a structure-based method, while that based on stable motifs takes into account the functional information of the network (network dynamics) and has a substantial improvement in computing the number of driver nodes. These two methods are very promising, even though none of them guarantees to find the minimal set of driver nodes.
In~\cite{L16}, Wang et al.\ highlighted an experimentally feasible approach towards the control of nonlinear dynamical networks by constructing `attractor networks' that reflect their controllability.
They construct the attractor network of a system by including all the experimentally validated paths between the attractors of the network.
The concept of an attractor network is very inspiring.
However, this method cannot provide a straightforward way to find the paths from one attractor to a desired attractor,
and it fails to formulate a generic mathematical framework for the control of nonlinear dynamical networks.
Other approaches taking into account the dynamic properties of non-linear BNs include Rocha et al.~\cite{MPR13,GR16} who explore the notion of canalisation and canalising
functions in BNs to reason about their dynamics and steady state behaviour.

Closely related to our work, Mandon et al.~\cite{MHP16,MHP17} proposed approaches towards the control of asynchronous BNs.
In particular, in~\cite{MHP16} they proposed a few algorithms to identify reprogramming determinants for both existential and
inevitable reachability of the target attractor with permanent perturbations. 
Later on, they proposed an algorithm that can find all existing control paths between two states
within a limited number of either permanent or temporary perturbations~\cite{MHP17}. 
However, these methods do not scale well for large networks.\footnote{We learnt through private communication that
the current implementation of their methods does not scale efficiently to BNs having more than 20 nodes}
This is mainly due to the fact that they need to encode all possible control strategies into the
transition system of the BN in order to identify the desired reprogramming paths~\cite{MHP17}.
As a consequence, the size of the resulting {\em perturbed transition graph} grows exponentially with the number of allowed perturbations,
which renders their algorithms inefficient.

The identified limitations of these existing approaches motivate us to develop a new approach towards the control of non-linear Boolean networks which is modular and
exploits {\em both} their structural and dynamic properties. Gates et al.~\cite{GR16} showed that such an approach is inevitable for the identification of
the correct parameters and control strategies, in that, focussing only on a single property (either structural or dynamic) might lead to both their
overestimation or underestimation.


\section{Preliminaries}
\label{sec:pre}

\subsection{Boolean networks}
A Boolean network (BN) describes elements of a dynamical system with binary-valued nodes and interactions between elements with Boolean functions. It is formally defined as:

\begin{definition}[Boolean networks]
  A Boolean network is a tuple $\bn = (\x,\f)$ where $\x=(x_1,x_2,\ldots, x_n)$ such that each $x_i, 1\leq i\leq n$ is a Boolean variable and $\f=(f_1,f_2,\ldots,f_n)$ is a tuple of Boolean functions over $\x$. $|\x| = n$ denotes the number of variables.
\end{definition}

In what follows, $i$ will always range between 1 and $n$, unless stated otherwise. A Boolean network $\bn=(\x,\f)$ may be viewed as a directed graph $\graph_\bn = (V,E)$ where $V=\{v_1,v_2\ldots, v_n\}$ is the set of {\sf vertices} or {\sf nodes} and for every $1\leq i,j\leq n$, there is a directed edge from $v_j$ to $v_i$ if and only if $f_i$ depends on $x_j$. An edge from $v_j$ to $v_i$ will be often denoted as $v_j\rightarrow v_i$. A {\sf path} from a vertex $v$ to a vertex $v'$ is a (possibly empty) sequence of edges from $v$ to $v'$ in $\graph_\bn$. For any vertex $v\in V$ we define its set of parents as $\prt(v)=\{v'\in V\ |\ v'\rightarrow v\}$. 
For the rest of the exposition, we assume that an arbitrary but fixed network $\bn$ of $n$ variables is given to us and $\graph_\bn=(V,E)$ is its associated directed graph.

A {\sf state} $\state$ of $\bn$ is an element in $\{0,1\}^n$. Let $\St$ be the set of states of $\bn$. For any state $\state=(s_1,s_2,\ldots,s_n)$, and for every $i$, the value of $s_i$, often denoted as $\state[i]$, represents the value that the variable $x_i$ takes when the $\bn$ `is in state $\state$'. For some $i$, suppose $f_i$ depends on $x_{i_1},x_{i_2},\ldots, x_{i_k}$. Then $f_i(\state)$ will denote the value $f_i(\state[i_1],\state[i_2],\ldots, \state[i_k])$. 
For two states $\state,\state'\in\St$, the {\sf Hamming distance} between $\state$ and $\state'$ will be denoted as $\hd(\state,\state')$. For a state $\state$ and a subset $\St'\subseteq\St$, the Hamming distance between $\state$ and $\St'$ is defined as $\hd(\state,\St')=\min_{\state'\in\St'}\hd(\state,\state')$. We let $\arg(\hd(\state,\St'))$ denote the set of subsets of $\{1,2,\ldots,n\}$ such that $I\in \arg(\hd(\state,\St'))$ if and only if $I$ is a set of indices of the variables that realise this Hamming distance.

\subsection{Dynamics of Boolean networks}\label{sec:dynamics}
We assume that the Boolean network evolves in discrete time steps. It starts initially in a state $\state_0$ and its state changes in every time step according to the update functions $\f$. The updating may happen in various ways. Every such way of updating gives rise to a different dynamics for the network. In this work, we shall be interested primarily in the asynchronous updating scheme.

\begin{definition}[Asynchronous dynamics of Boolean networks]\label{def:dynamics} Suppose $\state_0\in\St$ is an initial state of $\bn$. The asynchronous evolution of $\bn$ is a function $\update: \naturals \rightarrow \wp(\St)$ such that $\update(0)=\state_0$ and for every $j\geq 0$, if $\state\in\update(j)$ then $\state'\in \update(j+1)$ if and only if $\hd(\state,\state')\leq 1$ and there exists $i$ such that $\state'[i]=f_i(\state)$.
\end{definition}

Note that the asynchronous dynamics is non-deterministic -- the value of exactly one variable is updated in a single time-step. The index of the variable that is updated is not known in advance. Henceforth, when we talk about the dynamics of $\bn$, we shall mean the asynchronous dynamics as defined above.

The dynamics of a Boolean network can be represented as a {\em state transition graph} or a {\em transition system (TS)}.

\begin{definition}[Transition system of $\bn$]\label{def:ts}The transition system of $\bn$, denoted by the generic notation $\ts$ is a tuple $(\St,\rightarrow)$ where the vertices are the set of states $\St$ and for any two states $\state$ and $\state'$ there is a directed edge from $\state$ to $\state'$, denoted $\state\rightarrow\state'$ if and only if $\hd(\state,\state')\leq 1$ and there exists $i$ such that $\state'[i]=\f_i(\state)$.
\end{definition}

\subsection{Attractors and basins of attraction}
A {\sf path} from a state $\state$ to a state $\state'$ is a (possibly empty) sequence of transitions from $\state$ to $\state'$ in $\ts$. A path from a state $\state$ to a subset $\St'$ of $\St$ is a path from $\state$ to any state $\state'\in \St'$. For any state $\state\in \St$,
let $\pre_\ts(\state) = \{\state'\in \St\ |\ \state'\rightarrow \state\}$ and let $\post_\ts(\state) = \{\state'\in \St\ |\ \state\rightarrow \state'\}$.
$\pre_\ts(\state)$ contains all the states that can reach $\state$ by performing a single transition in $\ts$ and
$\post_\ts(s)$ contains all the states that can be reached from $\state$ by a single transition in $\ts$. Note that, by
definition, $\hd(\state,\pre_\ts(\state))\leq 1$ and $\hd(\state,\post_\ts(\state))\leq 1$. $\pre_\ts$ and $\post_\ts$
can be lifted to a subset $\St'$ of $\St$ as: $\pre_\ts(\St') = \bigcup_{\state\in\St'}\pre_\ts(\state)$ and
$\post_\ts(\St') = \bigcup_{\state\in\St'}\post_\ts(\state)$.

For a state $\state\in\St$, $\reach_\ts(\state)$ denotes the set of states $\state'$ such that there is a path from $\state$ to $\state'$ in $\ts$ and can be defined as the transitive closure of the $\post_\ts$ operation. Thus, $\reach_\ts(\state)$ is the smallest subset of states in $\St$ such that $\state\in\reach_\ts(\state)$ and $\post_\ts(\reach_\ts(\state))\subseteq\reach_\ts(\state)$.

\begin{definition}[Attractor]An attractor $A$ of $\ts$ (or of $\bn$) is a subset of states of $\St$ such that for every $\state\in A, \reach_\ts(\state)=A$.
\end{definition}

Any state which is not part of an attractor is a transient state. An attractor $A$ of $\ts$ is said to be reachable from a state $\state$ if $\reach_\ts(\state)\cap A\neq\emptyset$. Attractors represent the stable behaviour of the $\bn$ according to the dynamics. The network starting at any initial state $\state_0\in \St$ will eventually end up in one of the attractors of $\ts$ and remain there forever unless perturbed. The following is a straightforward observation.
\begin{observation}\label{obs:attr} \normalfont Any attractor of $\ts$ is a bottom strongly connected component of $\ts$.
\end{observation}

For an attractor $A$ of $\ts$, we define subsets of states of $\St$ called the weak and strong basins of attractions of $A$, denoted as $\bas^W_\ts(A)$ and $\bas^S_\ts(A)$ resp. as follows.

\begin{definition}[Basin of attraction]
  Let $A$ be an attractor of $\ts$.
  \begin{itemize}
  \item {\bf Weak basin:} The weak basin of attraction of $A$ with respect to $\ts$, is defined as $\bas^W_\ts(A) = \{\state\in\St\ |\ \reach_\ts(\state)\cap A\neq \emptyset\}$.
    \item {\bf Strong basin:} The strong basin of attraction of $A$ with respect to $\ts$, is defined as $\bas_\ts^S(A)= \bas_\ts^W(A)\setminus\bas_\ts^W(A')$ where $A'$ is an attractor of $\ts$ and $A'\neq A$.
  \end{itemize}
\end{definition}

Thus the weak basin of attraction of $A$ is the set of all states $\state$ from which there is a path to $A$. It is possible that there are paths from $\state$ to some other attractor $A'\neq A$. However, the notion of a strong basin does not allow this. Thus, if $\state\in\bas^S_\ts(A)$ then $\state\notin\bas^W_\ts(A')$ for any other attractor $A'$. We need the notion of strong basin to {\em ensure} reachability to the target attractor after applying control.

\begin{example}\label{eg:attrbas}
  Consider the three-node network $\bn=(\x,\f)$ where $\x=(x_1,x_2,x_3)$ and $\f=(f_1,f_2,f_3)$ where $f_1=\neg x_2 \lor (x_1\land x_2), f_2=x_1\land x_2$ and $f_3=x_3\land\neg (x_1\land x_2)$. The graph of the network $\graph_\bn$ and its associated transition system $\ts$ is given in Figure~\ref{fig:fullts}. $\ts$ has three attractors $\{(100)\}, \{(110)\}$ and $\{(101)\}$ shown shaded in pink. Their corresponding strong basins of attractions are shown by enclosing blue shaded regions. Note that for this particular example, both the strong and the weak basins are the same for all the attractors.

  \begin{figure}
\centering
\begin{tikzpicture}

\node        (s1)                  {$~v_1~$};
\node(s2)  [right=of s1]   {$~v_2~$};
 \node        (s3)  [right=of s2]   {$~v_3~$};

 \path     
 (s1)  edge [<-,loop above] node {} (s1)
 (s1)  edge [->,bend left] node {} (s2)

 (s2)  edge [<-,loop above] node {} (s2)
 (s2)  edge [->] node {} (s1)
 (s2)  edge [->] node {} (s3)
  (s1)  edge [->,bend right] node {} (s3)
(s3)  edge [<-,loop above] node {} (s3)
 ;
\end{tikzpicture}

{\begin{tikzpicture}
[attr/.style={draw=none,fill=red!50,rounded corners}]
\draw[draw=none,fill=blue!50,rounded corners,opacity=0.5] (0,1.4) -- (2,1.4) -- (2,0.6) -- (0.5,0.6) -- (0.5,-0.4) -- (-0.5,-0.4) -- (-0.5,1.4) -- (0,1.4);
\draw[draw=none,fill=blue!50,rounded corners,opacity=0.5] (1.5,0.4) -- (3.4,0.4) -- (3.4,-0.4) -- (1,-0.4) -- (1,0.4) -- (1.5,0.4);
\draw[draw=none,fill=blue!50,rounded corners,opacity=0.5] (3,1.4) -- (5,1.4) -- (5,-0.4) -- (4,-0.4) -- (4,0.6) -- (2.5,0.6) -- (2.5,1.4) -- (3,1.4);
\node (t1) at (0,1) {000};
\node (t2) at (1.5,1) {010};
\node (t3) at (3,1) {011};
\node (t4) at (4.5,1) {001};
\node (t5)[attr] at (0,0) {100};
\node (t6)[attr] at (1.5,0) {110};
\node (t7) at (3,0) {111};
\node (t8)[attr] at (4.5,0) {101};

\path
(t1) edge [<-,loop above] (t1)
(t2) edge [<-,loop above] (t2)
(t3) edge [<-,loop above] (t3)
     (t4) edge [<-,loop above] (t4)
    (t5) edge [<-,loop below,min distance=6mm,in=-65,out=-115,looseness=10] (t5)
(t6) edge [<-,loop below,min distance=6mm,in=-65,out=-115,looseness=10] (t6)
(t7) edge [<-,loop below,min distance=6mm,in=-65,out=-115,looseness=10] (t7)
(t8) edge [<-,loop below,min distance=6mm,in=-65,out=-115,looseness=10] (t8)
(t1) edge [->] (t5)
(t2) edge [->] (t1)
(t3) edge [->] (t4)
(t4) edge [->] (t8)
(t7) edge [->] (t6)
;
\end{tikzpicture}}
\caption{the graph of $\bn$ and its transition system}
\label{fig:fullts}%
\end{figure}
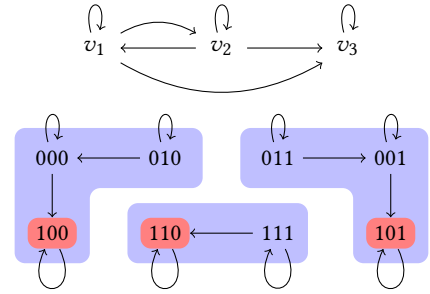
\end{example}

\begin{observation}\label{obs:weakbasin}
\normalfont  Given an attractor $A$, we can compute the weak basin $\bas_\ts^W(A)$ by a simple iterative fixpoint procedure. Indeed, $\bas_\ts^W(A)$ is the smallest subset $W$ of $\St$ such that $A\in W$ and $\pre_\ts(W)\subseteq W$. We shall call this procedure {\sc Compute\_Weak\_Basin} which will take as arguments the function tuple $\f$ and an attractor $A$.
\end{observation}

Henceforth, to avoid clutter, we shall drop the subscript $\ts$ when the transition system is clear from the context. Also, we shall often drop the superscript $S$ as well the mention of the word ``strong'' when dealing with strong basins. Thus the ``basin of $A$'' will always mean the strong basin of attraction of $A$ unless mentioned otherwise and will be denoted as $\bas(A)$.

\subsection{The control problem}\label{ssec:controlProblem}
As described in the introduction, the attractors of a Boolean network represent the cellular phenotypes, the expressions of the genes etc. Some of these attractors may be diseased, weak or undesirable while others are healthy and desirable. Curing a disease is thus in effect, moving the dynamics of the network from an undesired `source' attractor to a desired `target' attractor.

One of the ways to achieve the above is by controlling the various `parameters' of the network, for eg. the values of the variables, or the Boolean functions themselves. In this exposition, we shall be interested in the former kind of control, that is, tweaking the values of the variables of the network. Such a control may be (i) {\it permanent} -- the value(s) of one or more variables are fixed forever, for all the following time steps or (ii) {\it temporary} -- the values of (some of) the variables are fixed for a finite number (one or more) of time steps and then the control is removed to let the system evolve on its own. Moreover, the variables can be either controlled (a) {\it simultaneously} -- the control is applied to all the variables at once or (b) {\it sequentially} -- the control is applied over a sequence of steps. 

In this work we shall be interested in the control of type (ii) and (a). Moreover, for us, the perturbations are applied only for a {\em single} time step. Thus we can formally define control as follows.

\begin{definition}[Control]\label{def:control}
  A control $\control$ is a (possibly empty) subset of $\{1,2,\ldots, n\}$. For a state $\state\in \St$, the application of a control $\control$ to $\state$, denoted $\control(\state)$ is defined as the state $\state'\in \St$ such that $\state'[i]=(1-\state[i])$ if $i\in \control$ and $\state'[i]=\state[i]$ otherwise. Given a control $\control$, the set of vertices $\{v_i\ |\ i\in \control\}$ of $\graph_\bn$ will be called the driver nodes for $\control$.
\end{definition}

Our aim is to make the control as less invasive to the system as possible. Thus not only is the control applied for just a single time step, it is also applied to as few of the nodes of the Boolean network as possible. The minimal simultaneous single-step target-control problem for Boolean networks that we are thus interested in can be formally stated as follows.

\smallskip
\noindent{\bf Minimal simultaneous target-control:} Given a Boolean network $\bn$, a `source state' $\state\in\St$ and a `target attractor' $A_t$ of $\ts$, compute a control $\control$ such that after the application of $\control(\state)$, $\bn$ eventually reaches $A_t$ and $\control$ is a minimal such subset of $\{1,2,\ldots, n\}$. We shall call such a control a {\em minimal control} from $\state$ to $A_t$. The set of all minimal controls from $\state$ to $A_t$ will be denoted as $\cset_{\min}^{\state\rightarrow A_t}$.

\smallskip
Note that the requirement of minimality is crucial, without which the problem is rendered trivial - simply pick some state $\state'\in A_t$ and move to it. The nodes required to be controlled will often be called the {\sf driver nodes} for the corresponding control. Our goal is to provide an efficient algorithm for the above question. That is, to devise an algorithm that takes as input only the Boolean functions $\f$ of $\bn$, a source state $\state$ and a target attractor $A_t$ of $\ts$ and outputs the indices of a minimal subset of nodes of $\state$ that need to be toggled or controlled (the driver nodes) so that after applying the control, the dynamics eventually and surely reaches $A_t$. It is known that in general the problem is computationally difficult -- PSPACE-hard \cite{MHP16} and unless certain open conjectures in computational complexity are false, these questions are computationally difficult and would require time exponential in the size of the Boolean network. That is intuitively because of the infamous state-space explosion phenomenon -- the number of states of the transition system is exponential in the network-size.

\begin{observation}\label{obs:mincontrol}
  \normalfont  It is important to note that if the BN is in some state $\state\in \bas(A)$ in some time step $t$, that is if $\xi(t)=\state$ then by the definition of $\bas(A)$ it will eventually and surely reach a state $\state'\in A$. That is, there exists a time step $t'>t$ such that $\xi(t')=\state'$. Hence given a source state $\state$ and a target attractor $A_t$, $\cset_{\min}^{\state\rightarrow A_t}$ can easily be seen to be equal to $\arg(\hd(\state,\bas(A_t)))$. In other words
  \begin{proposition}\label{prop:mincontrol}
    A control $\control$ from $\state$ to $A_t$ is minimal if and only if $\control(\state)\in \bas(A_t)$ and $\control\in \arg(\hd(\state,\bas(A)))$.
  \end{proposition}

  \begin{proof}
    Indeed, since if $\control(\state)\notin\bas(A_t)$ then $\bn$ is not assured to reach a state in $A_t$ or if $\control\notin\arg(\hd(\state,\bas(A)))$ then $C$ cannot be minimal, and conversely.
  \end{proof}
  
  Thus, solving the minimal simultaneous target-control problem efficiently boils down to how efficiently we can compute the strong basin of the target attractor.
\end{observation}

\begin{example}\label{eg:egcontd}
  Continuing with Example \ref{eg:attrbas}, suppose we are in source state $\state=(101)$ (which is also an attractor) and we want to apply (minimal simultaneous) control to $\state$ so the system eventually and surely moves to the target attractor $A_t=\{(110)\}$. We could flip $\state[2]$ and $\state[3]$ to move directly to $A_t$ which would require a control $\control = \{2,3\}$. However, if we notice that the state $(111)$ is in the basin of $A_t$ we can simply apply a control $\control'=\{2\}$ and the dynamics of the $\bn$ will ensure that it eventually reaches $A_t$. Indeed, $\control'$ is also the minimal control in this case.
\end{example}

\subsection{A global algorithm}
In the rest of this section, we first describe a procedure for computing the (strong) basin 
of an attractor based on the computation of fixed point. We then use this procedure to design a simple global algorithm for solving the minimal simultaneous target-control problem based on a global computation of the basin of the target attractor $A_t$. This algorithm will act as a reference for comparing the decomposition-based algorithm which we shall later develop.

We first introduce an algorithm called {\sc Compute\_Strong\_Basin}, described in Algorithm~\ref{alg:stb}, for the computation of the strong basin of an attractor $A$ based on a fixpoint approach. 
We shall use this algorithm in both the global minimal control algorithm and later in the decomposition-based algorithm.
A proof of correctness of Algorithm~\ref{alg:stb} can be found in the appendix.
\begin{algorithm}[!t]
\caption{Fixpoint computation of strong basin}
\label{alg:stb}
\begin{algorithmic}[1]
  \Procedure{Compute\_Srong\_Basin}{$\f$,$A$}
  \State Let {\sc WB = Compute\_Weak\_Basin($\f,A$)}
  \State Initialise {\sc SB $=\emptyset$}
  \State Till {\sc SB $\neq$ WB} do
  \State \qquad If {\sf SB}$\neq \emptyset$ do {\sf WB} = {\sf SB} \label{st:fp}
  \State \qquad {\sc SB=WB}$\setminus(\pre(\post(${\sc WB}$)\setminus \wb)\cap${\sc WB}$)$
  \State done
  \State \Return {\sc SB}
  \EndProcedure
\end{algorithmic}
\end{algorithm}

\begin{algorithm}[!t]
\caption{Global minimal simultaneous target control}
\label{alg:controlglobal}
\begin{algorithmic}[1]
  \Procedure{Global\_Minimal\_control}{$\f,\state,A_t$}
  \State Let {\sc SB = Compute\_Strong\_Basin($\f,A_t$)}
  \State \Return $\arg(\hd(\state,${\sc SB}$))$
  \EndProcedure
\end{algorithmic}
\end{algorithm}

We now use the algorithm {\sc Compute\_Strong\_Basin} to give a global algorithm, Algorithm~\ref{alg:controlglobal}, for the minimal simultaneous target control problem.
Note that Algorithm~\ref{alg:controlglobal} is worst-case exponential in the size of the input (the description of $\bn$). Indeed, since the basin of attraction of $A_t$ might well be equal to all the states of the entire transition system $\ts$ which is exponential in the description of $\bn$. Now, although an efficient algorithm for this problem is highly unlikely, it is possible that when the network has a certain well-behaved structure, one can do better than this global approach. Most of the previous attempts at providing such an algorithm for such well-behaved networks either exploited exclusively the structure of the network or failed to minimise the number of driver nodes. Here we show that, when we take both the structure and the dynamics into account, we can have an algorithm which, for certain networks, is much more efficient than the global approach.


\section{A Decomposition-based Approach}
\label{sec:approach}
Note that our global solution for the minimal control problem, Algorithm~\ref{alg:controlglobal}, is generic, in that, we can plug into it any other algorithm for computing the basin of the target attractor and it would still work. Its performance, however, directly depends on the performance of the particular algorithm used to compute this basin.

In this section, we demonstrate an approach to compute the basin of attraction of $A_t$ based on the decomposition of the BN into structural components called {\em blocks}. This will then be used to solve the minimal control problem. The approach is based on that of~\cite{MPY17b} for computing the attractors of asynchronous Boolean networks. The overall idea is as follows. The network is divided into {\em blocks} based on its strongly connected components. The blocks are then sorted topologically resulting in a dependency graph of the blocks which is a directed acyclic graph (DAG). The transition systems of the blocks are computed inductively in the sorted order and the target attractor $A_t$ is then projected to these blocks. The local strong basins for each of these projections are computed in the transition system of the particular block. These local basins are then combined to compute the global basin $\bas(A_t)$.

\subsection{Blocks}
Let $\scc$ denote the set of maximal strongly connected components (SCCs) of $\graph_\bn$.\footnote{By convention, we assume that a single vertex (with or without a self loop) is always an SCC, although it may not be maximal.} Let $W$ be an SCC of $\graph_\bn$. The set of parents of $W$ is defined as $\prt(W)=(\bigcup_{v\in W}\prt(v))\setminus W$.

\begin{definition}[Basic Block]
A basic block $B$ is a subset of the vertices of $V$ such that $B= W\cup\prt(W)$ for some $W\in \scc$.
\end{definition}

Let $\blocks$ be the set of basic blocks of $\graph_\bn$. Since every vertex of $\graph_\bn$ is part of an SCC, we have $\bigcup\blocks = V$. The union of two or more basic blocks of $\blocks$ will also be called a {\em block}. For any block $B$, $|B|$ will denote the number of vertices in $B$. Using the set of basic blocks $\blocks$ as vertices, we can form a directed graph $\graph_\blocks = (\blocks, E_\blocks)$, which we shall call the {\em block graph} of $\bn$. The vertices of $\graph_\blocks$ are the basic blocks and for any pair of basic blocks $B',B\in \blocks, B'\neq B$, there is a directed edge from $B'$ to $B$ if and only if $B'\cap B\neq\emptyset$ and for every $v\in (B'\cap B)$, $\prt(v)\cap B =\emptyset$. In such a case, $B'$ is called a {\em parent} block of $B$ and $v$ is called a {\em control node} for $B$. Let $\prt(B)$ and $\ctr(B)$ denote the set of parent blocks and the set of control nodes of $B$ resp. It is easy to observe that

\begin{observation}
$\graph_\blocks$ is a directed acyclic graph (DAG).
\end{observation}

A block $B$ (basic or non-basic) is called {\em elementary} if $\prt(v)\subseteq B$ for every $v\in B$. $B$ is called {\em non-elementary} otherwise. We shall henceforth assume that $\bn$ has $k$ basic blocks and they are topologically sorted as $\{B_1,B_2,\ldots,B_k\}$. Note that for every $j:1\leq j\leq k$, $(\bigcup_{\ell=1}^j B_\ell)$ is an elementary block. We shall denote it as $\B_j$.

For two basic blocks $B$ and $B'$ where $B$ is non-elementary, $B'$ is said to be an {\sf ancestor} of $B$ if there is a path from $B'$ to $B$ in the block graph $\graph_\blocks$. The {\sf ancestor-closure} of a basic block $B$ (elementary or non-elementary), denoted $\ac(B)$ is defined as the union of $B$ with all its ancestors. Note that $\ac(B)$ is an elementary block and so is $\{\ac(B')\ |\ B'\in\prt(B)\}$, which we denote as $\ac(B)^-$.


\subsection{Projection of states and the cross operation}
We shall assume that the vertices $\{v_1,v_2,\ldots, v_n\}$ of $\graph_\bn$ inherit the ordering of the variables $\x$ of $\bn$. Let $B$ be a block of $\bn$. Since $B$ is a subset of $V$ its state space is $\{0,1\}^{|B|}$ and is denoted as $\St_B$. For any state $\state\in \St$, where $\state=(s_1,s_2,\ldots, s_n)$, the projection of $\state$ to $B$, denoted $\state|_B$ is the tuple obtained from $\state$ by suppressing the values of the variables not in $B$. Thus if $B=\{v_{i_1}, v_{i_2},\ldots, v_{i_k}\}$ then $\state|_{B} = (s_{i_1},s_{i_2},\ldots, s_{i_k})$. Clearly $\state|_B\in\St_B$. For a subset $\St'$ of $\St$, $\St'|_B$ is defined as $\{\state|_B\ |\ \state\in \St'\}$.

\begin{definition}[Cross Operation] Let $B_1$ and $B_2$ be two blocks of $\bn$ and let $\state_1$ and $\state_2$ be states of $B_1$ and $B_2$ resp. $\state_1\cross\state_2$ is defined (called {\em crossable}) if there exists a state $\state\in \St_{B_1\cup B_2}$ such that $\state|_{B_1}=\state_1$ and $\state|_{B_2}=\state_2$. $\state_1\cross\state_2$ is then defined to be this unique state $\state$. For any subsets $\St_1$ and $\St_2$ of $\St_{B_1}$ and $\St_{B_2}$ resp. $\St_1\cross\St_2$ is a subset of $\St_{B_1\cup B_2}$ and is defined as:
$$\St_1\cross\St_2 = \{\state_1\cross\state_2\ |\ \state_1\in\St_1, \state_2\in\St_2 \text{ and } \state_1 \text{ and } \state_2 \text{ are crossable}\}$$
\end{definition}

Note that $\St_1\cross \St_2$ can be the empty set. The cross operation is easily seen to be associative. Hence for more than two states $\state_1,\state_2,\ldots,\state_k$, $\state_1\cross\state_2\cross\ldots\state_k$ can be defined as $(((\state_1\cross\state_2)\cross\ldots)\cross\state_k)$. We have a similar definition for the cross operation on more than two sets of states.

\subsection{Transition system of the blocks}\label{sec:tsblocks}
The next step is to describe how to construct the `local' transition systems of each of the blocks. These transition systems will be inductively defined starting from the elementary blocks and moving to the blocks further down the topological order. For an elementary block $B$ (basic or non-basic), its transition system $\ts_B$ is given exactly as Definition~\ref{def:ts} with the vertices being $\St_B$. This is well-defined since by the definition of an elementary block, the update functions of the vertices of $B$ do not depend on the value of any vertex outside $B$. On the other hand, the transition system of a non-elementary block $B$ depends on the transitions of its parent blocks (or its control nodes in its parent blocks). The transition system of such a block thus has to be defined based on (some or all of) the transitions of its parent blocks.

Towards that let $B$ be a non-elementary basic block of $\bn$ and let $A$ be an attractor of the transition system of the elementary block $\ac(B)^-$ and let $\bas(A)$ be its (strong) basin of attraction. Then

\begin{definition}[TS of non-elementary blocks]\label{def:tsblocks} The transition system of $B$ generated by $\bas(A)$ is defined as a tuple $\ts_B=(\St,\rightarrow)$ where the set of states $\St$ of $\ts_B$ is a subset of $\St_{\ac(B)}$ such that $\state\in \St$ if and only if $\state|_{\ac(B)^-}\in \bas(A)$ and for any two states $\state,\state'\in\St_{\ac(B)}$ there is a transition $\state\rightarrow\state'$ if and only if $\hd(\state,\state')\leq 1$ and there exists $i$ among the indices of the nodes in $\ac(B)$ such that $\state'[i]=f_i(\state)$.
\end{definition}

\noindent{\bf Remark.} Our construction of the transition system of the non-elementary blocks is different from that used in~\cite{MPY17b}. There, for a non-elementary block $B$, the set of states of $\ts_B$ was a subset of $\St_B$ and the transitions for the control nodes of $B$ were derived by projecting the transitions in the attractor of the parent block of $B$ to these control nodes. It can be shown that such an approach does not work for the decomposition-based solution to the minimal simultaneous target-control problem that we aim for here and we need the full behaviour of the basin of the attractor of the parent blocks of $B$ to generate the transition system of $B$.

\subsection{The main results}\label{sec:results}
We now give the key results of the above constructions which will form the basis of the decomposition-based control algorithm that we shall develop in the next section. To maintain the continuity and flow of the main text, we shall defer all the proofs to Appendix~\ref{appendix:proofs}.

Suppose $\bn$ has $k$ blocks which are topologically ordered as $\{B_1,B_2,\ldots, B_k\}$. Let $\ts$ be the transition system of $\bn$ and for every attractor $A$ of $\ts$ and for every $j:1\leq j\leq k$ let $A_j=A|_{B_j}$ be the projection of $A$ to $B_j$. We then have

\begin{theorem}[Preservation of attractors]\label{thm:attpres} Suppose for every attractor $A$ of $\ts$ and for every $i:1\leq i<k$, if $B_{i+1}$ is non-elementary then $\ts_{i+1}$ is realized by $\bas(\cross_{j\in I}A_j)$, its basin w.r.t. the transition system for $(\bigcup_{j\in I}B_j)$, where $I$ is the set of indices of the basic blocks in $\ac(B_{i+1})^-$. We then have, for every $i:1\leq i<k$, $A_{i+1}$ is an attractor of $\ts_{i+1}$,  $(\cross_{j\in I}A_j\cross A_{i+1})$ is an attractor of the transition system for the elementary block $(\bigcup_{j\in I}B_j\cup B_{i+1})$, $(\cross_{j=1}^{i+1}A_j)$ is an attractor of the transition system $\tshat_{i+1}$ of $\B_{i+1}$ and $A$ is an attractor of $\ts_k$.
\end{theorem}

\begin{theorem}[Preservation of basins]\label{thm:bassubset} Given the hypothesis and the notations of Theorem~\ref{thm:attpres}, we have $(\cross_{i\leq k}\bas(A_i))= \bas(A)$ where $\bas(A)$ is the basin of attraction of the attractor $A=(A_1\cross A_2\cross \ldots\cross A_k)$ of $\ts$.
\end{theorem}

\begin{example}\label{eg:egcontd2}
Continuing with Example \ref{eg:attrbas} and \ref{eg:egcontd}, we note that $\bn$ has two maximal SCCs $\{v_1,v_2\}$ and $\{v_3\}$. These give rise to two blocks $B_1=\{v_1,v_2\}$ and $B_2=\{v_1,v_2,v_3\}$ shown in Figure \ref{fig:blocks}. $B_1$ is elementary whereas $B_2$ is non-elementary where $B_1$ is its parent and it has control nodes $v_1$ and $v_2$.

\begin{figure}
\centering
\begin{tikzpicture}
\draw[draw=none,fill=blue!40,opacity=0.5] (1,0) ellipse (1.5cm and 0.7cm);
  \draw[draw=none,fill=blue!40,opacity=0.5] (2,0) ellipse (2.5cm and 1cm);

\node        (s1)                  {$~v_1~$};
\node(s2)  [right=of s1]   {$~v_2~$};
 \node        (s3)  [right=of s2]   {$~v_3~$};

 \path     
 (s1)  edge [<-,loop above] node {} (s1)
 (s1)  edge [->,bend left] node {} (s2)

 (s2)  edge [<-,loop above] node {} (s2)
 (s2)  edge [->] node {} (s1)
 (s2)  edge [->] node {} (s3)
  (s1)  edge [->,bend right] node {} (s3)
(s3)  edge [<-,loop above] node {} (s3)
 ;
\end{tikzpicture}
\caption{The blocks of $\bn$}
\label{fig:blocks}
\end{figure}
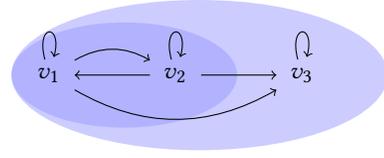

The transition system of block $B_1$ is shown in Figure \ref{fig:tsblocks}(a). It has two attractors $\{(10)\}$ and $\{(11)\}$ shown in pink with their corresponding strong basins shown in shaded blue regions. The transision system of the block $B_2$ generated by the basin of the attractor $\{(10)\}$ of the block $B_1$ is shown in Figure \ref{fig:tsblocks}(b). It has two attractors $\{(100)\}$ and $\{(101)\}$ shown again in pink with their corresponding basins of attractions shown in blue. Note that, indeed, according to Theorem \ref{thm:attpres} we have that $\{(10)\}\cross\{(100)\} = \{(100)\}$ and $\{(10)\}\cross\{(101)\} = \{(101)\}$ are attractors of the global transition system of $\bn$. Also note that taking the cross of the local basins of attractions does indeed result in the global basins.

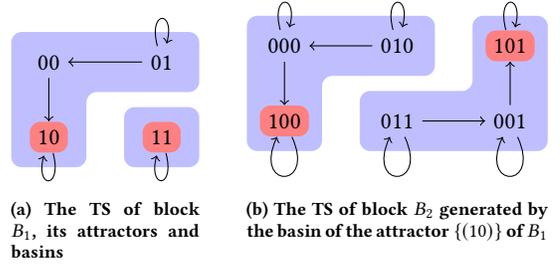
\begin{figure}
\centering
\subfloat[The TS of block $B_1$, its attractors and basins]{%
\begin{tikzpicture}
[attr/.style={draw=none,fill=red!50,rounded corners}]
\draw[draw=none,fill=blue!50,rounded corners,opacity=0.5] (0,1.4) -- (2,1.4) -- (2,0.6) -- (0.5,0.6) -- (0.5,-0.4) -- (-0.5,-0.4) -- (-0.5,1.4) -- (0,1.4);
\draw[draw=none,fill=blue!50,rounded corners,opacity=0.5] (1.5,-0.4) -- (2,-0.4) -- (2,0.4) -- (1,0.4) -- (1,-0.4) -- (1.5,-0.4);
\node (t1) at (0,1) {00};
\node (t2) at (1.5,1) {01};
\node (t3)[attr] at (0,0) {10};
\node (t4)[attr] at (1.5,0) {11};

\path
(t2) edge [<-,loop above] (t2)
(t3) edge [<-,loop below] (t3)
(t4) edge [<-,loop below] (t4)
(t1) edge [->] (t3)
(t2) edge [->] (t1);
\end{tikzpicture}}%
\qquad
\subfloat[The TS of block $B_2$ generated by the basin of the attractor $\{(10)\}$ of $B_1$]{%
\begin{tikzpicture}
[attr/.style={draw=none,fill=red!50,rounded corners}]
\draw[draw=none,fill=blue!50,rounded corners,opacity=0.5] (0,1.4) -- (2,1.4) -- (2,0.6) -- (0.5,0.6) -- (0.5,-0.4) -- (-0.5,-0.4) -- (-0.5,1.4) -- (0,1.4);
\draw[draw=none,fill=blue!50,rounded corners,opacity=0.5] (1.5,0.4) -- (2.5,0.4) -- (2.5,1.4) -- (3.5,1.4) -- (3.5,-0.4) -- (1,-0.4) -- (1,0.4) -- (1.5,0.4);
\node (t1) at (0,1) {000};
\node (t2) at (1.5,1) {010};
\node (t3)[attr] at (3,1) {101};
\node (t5)[attr] at (0,0) {100};
\node (t6) at (1.5,0) {011};
\node (t7) at (3,0) {001};

\path
(t1) edge [<-,loop above] (t1)
(t2) edge [<-,loop above] (t2)
(t3) edge [<-,loop above] (t3)

    (t5) edge [<-,loop below,min distance=6mm,in=-65,out=-115,looseness=10] (t5)
(t6) edge [<-,loop below,min distance=6mm,in=-65,out=-115,looseness=10] (t6)
(t7) edge [<-,loop below,min distance=6mm,in=-65,out=-115,looseness=10] (t7)

(t1) edge [->] (t5)
(t2) edge [->] (t1)
(t6) edge [->] (t7)
(t7) edge [->] (t3)
;
\end{tikzpicture}}%
\caption{The local transition systems of the blocks $B_1$ and $B_2$}
\label{fig:tsblocks}
\end{figure}

\end{example}

\subsection{The decomposition-based algorithm}
Equipped with the results in Theorems~\ref{thm:attpres} and~\ref{thm:bassubset},
we can describe our procedure for computing the strong basin of the target attractor based on decomposing the BN into smaller blocks.
We shall later use this procedure to give an algorithm for the minimal control problem.
Towards that, Theorem~\ref{thm:bassubset} tells us that in order to compute $\bas(A_t)$
it is sufficient to compute the local basins of the projection of $A_t$ to each block $B_i$
(which by Theorem~\ref{thm:attpres} is an attractor of $B_i$) and finally merge these local basins using the cross operation.

Algorithm~\ref{alg:stblocal} implements this idea in pseudo-code.
It takes as input the graph $\graph_\bn$ and the update functions $\f$ of a given Boolean network,
and an attractor $A$ and returns the strong basin of attraction of $A$.
Line 2 decomposes $\graph_\bn$ into the blocks $\blocks$ (resulting in $k$ blocks) using the procedure {\sc Form\_Block} from~\cite{MPY17b}
and line 3 topologically sorts the blocks by constructing the block graph $\graph_\blocks$.
Lines 5-7 decomposes the attractor $A$ into its projection to the blocks.
Lines 8-17 then cycles through the blocks of $\blocks$ in topological order and for each block $B_i$:
if $B_i$ is elementary then constructs its transition system $\ts_i$ independently or,
if $B_i$ is non-elementary it constructs $\ts_i$ realised by the basin of $(A_1\cross A_2\cross\ldots\cross A_{i-1})$
which by Theorem~\ref{thm:attpres} is an attractor of $\ts_{i-1}$, the transition system for the elementary (non-basic) block $\B_{i-1}$.
Thus at every iteration $i$ of the for-loop the invariant that $A_i$ is an attractor of $\ts_i$ is maintained.
The procedure {\sc Compute\_Strong\_Basin}($\f|_{\overline{B_i}},A_i$) (lines 11,14), described in Algorithm~\ref{alg:stb},
computes the strong basin of $A_i$ w.r.t $\ts_i$. Line 16 extends the global strong basin $\SB$ computed so far
by crossing it with the local basin computed at each step. At the end of the for-loop $\SB$ will thus be equal to the global basin (by Theorem~\ref{thm:bassubset}).
It then easily follows that
\begin{proposition}
Algorithm~\ref{alg:stblocal} correctly computes the strong basin of the attractor $A$.
\end{proposition}

\begin{algorithm}[!t]
\caption{A decomposition-based procedure for the computation of strong basin}
\label{alg:stblocal}
\begin{algorithmic}[1]
\Procedure{Compute\_Strong\_Basin\_Decomp}{$\graph_\bn$,$\f$,$A$}
\State \label{formblock} $\blocks:=~${\sc Form\_Block}($\graph_\bn$);
\State $\blocks:=~${\sc Top\_Sort}($\blocks$);
\State $k:=$ size of $\blocks$; $\SB=\phi$; $\SB_i=\emptyset$; \hfil{\it //for all i}
\For {$i=1$ to $k$}
\State $A_i:=${\sc Decompose}($A,B_i$); \hfil{\it //Decompose the target attractor into block $B_i$}
\EndFor
\For {$i:=1$ to $k$} 
\label{line:forloopstart}
\If{$B_i$ is an~elementary block}
\label{line:detectblockstart}
\State $\ts_i:=$ transition system of $B_i$;  
\State $\SB_i:=${\sc Compute\_Strong\_Basin}$(\f|_{\overline{B_i}},A_i)$;
\Else {}
\State $\ts_i:=$ transition system of $B_i$ based on the basin of $(\cross_{j<i}A_j)$ in $\ts_{i-1}$;  
\State $\SB_i:=${\sc Compute\_Strong\_Basin}$(\f|_{\overline{B_i}},A_i)$;
\EndIf 
\State $\SB=${\sc Cross} $(\SB,\SB_i)$;
\EndFor
\State \Return $\SB$
\EndProcedure 

\end{algorithmic}
\end{algorithm}

We now plug the procedure {\sc Compute\_Strong\_Basin\_Decomp} of Algorithm~\ref{alg:stblocal} into Algorithm~\ref{alg:controlglobal}
to derive our decomposition-based minimal target control algorithm, Algorithm~\ref{alg:controllocal}, from source state $\state$ to target attractor $A_t$.

\begin{algorithm}[!t]
\caption{Decomposition-based minimal simultaneous target control}
\label{alg:controllocal}
\begin{algorithmic}[1]
  \Procedure{Decomp\_Minimal\_control}{$\graph_\bn,\f,\state,A_t$}
  \State Let {\sc SB = Compute\_Strong\_Basin\_Decomp($\graph_\bn,\f,A_t$)}
  \State \Return $\arg(s\hd(\state,$\SB$))$
  \EndProcedure
\end{algorithmic}
\end{algorithm}



\section{Case Studies}
\label{sec:casestudies}

To demonstrate the correctness and efficiency of our control framework,
we compare our decomposition-based approach with the global approach on both real-life biological networks and randomly generated networks.
Note that we do not compare our approach with the works by Mandon et al.~\cite{MHP16,MHP17},
as we are informed by the authors, through personal communication, that currently their methods cannot deal with networks larger than around 20 nodes.
The global approach and the decomposition-based approach, described by Algorithm~\ref{alg:controlglobal} and Algorithm~\ref{alg:controllocal}, are implemented in the software tool ASSA-PBN~\cite{MPY16b},
which is based on the model checker~\cite{MCMAS} to encode BNs into the efficient data structure binary decision diagrams (BDDs).
All the experiments are performed on a high-performance computing (HPC) platform, which contains CPUs of Intel Xeon X5675@3.07 GHz. 

\subsection{Case studies on biological networks}
\smallskip
\noindent{\bf The PC12 cell differentiation network} was developed by Offermann et al.~\cite{OKB16}.
It is a comprehensive model used to clarify the cellular decisions towards proliferation or differentiation.
It combines the temporal sequence of protein signalling, transcriptional response and subsequent autocrine feedback.
The model shows the interactions between protein signalling, transcription factor activity and
gene regulatory feedback in the regulation of PC12 cell differentiation after the stimulation of NGF.
Notice that the PC12 cell network is simulated in synchronous mode in~\cite{OKB16}.
In this paper, we treat the networks in asynchronous mode, as per Definition~\ref{def:dynamics}.
The BN model of the PC12 cell network consists of $32$ nodes and it has $7$ single-state attractors.
The network structure is divided into $19$ blocks by our decomposition approach (the procedure {\sc Form\_Block} in Algorithm~\ref{alg:stblocal}).
Details on the attractors and the decomposition of the network can be found in Appendix~\ref{app.casestudy}.

\smallskip
\noindent{\bf The apoptosis network} was constructed by Schlatter et al.~\cite{SSVSSBEMS09} based on extensive literature research.
Apoptosis is a kind of programmed cell death, the malfunction of which has been linked to many diseases.
In~\cite{SSVSSBEMS09}, they took into consideration the survival and metabolic insulin pathways,
the intrinsic and extrinsic apoptotic pathways and their crosstalks to build the Boolean network,
which simulates apoptotic signal transduction pathways with regards to different input stimulus. 
The BN model of this apoptosis network comprises 97 nodes and can be decomposed into 60 blocks by our decomposition approach (the procedure {\sc Form\_Block} in Algorithm~\ref{alg:stblocal}).
Using the asynchronous updating mode of BNs [Definition~\ref{def:dynamics}], 16 single-state attractors are detected when the housekeeping node
is set to on and six nodes (FASL, FASL\_2, IL\_1,TNF, UV, UV\_2) are set to false.
Details on the network structure and the decomposition are given in Appendix~\ref{app.casestudy} as well.

For the PC12 cell network and the apoptosis network, we aim to compute a minimal control $\control$ that can realise the minimal simultaneous single-step target control as explained in Section~\ref{ssec:controlProblem}.
That is to say, we compute the minimal set of driver nodes,
whose simultaneous single-step control can drive the network from a source state to a target attractor.
Since the attractors of the two networks are all single-state attractor, any of them can be taken as a source state.
All possible combinations of source and target attractors of the networks are explored and each case is repeated 100 times. 
The Hamming distances between attractors and the number of driver nodes for all cases are summarised in Table~\ref{tab:dr-pc12} and Table~\ref{tab:dr-apop}.
The attractors are labelled with numbers. The numbers in the first column and the first row represent the source and target attractors, respectively.
For each combination of source and target attractors, we list its Hamming distance (HD) and the number of driver nodes (\#D).
The numbers of driver nodes computed by the global and our decomposition-based approaches are identical,
demonstrating the correctness of our decomposition-based approach.
The \#D represents the results of both approaches. 

Table~\ref{tab:dr-pc12} and Table~\ref{tab:dr-apop} show that compared to the size of the network and the Hamming distance between the source and target attractors,
the minimal set of driver nodes required is quite small.
Especially for the apoptosis network with 97 nodes, the numbers of driver nodes are less than or equal to 4 for all the cases.
The PC12 cell network always reaches the same steady state with "cell differentiation" set to on by setting NGF to `on'~\cite{OKB16}.
To drive the network from any other attractor to this steady state, only NGF is required, which also shows the outstanding role of NGF in the network.

The speedups gained by our decomposition-based approach for different combinations of source and target attractors of the two networks are shown in Table~\ref{tab:sp-pc12} and Table~\ref{tab:sp-apop}.
\footnote{More details can be found in Appendix~\ref{app.casestudy}.}
For each case, the speedup is calculated with the formula $\textrm{speedup}=\frac{t_{\it global}}{t_{\it decom}}$,
where $t_{\it global}$ and $t_{\it decom}$ are the time costs of the global approach and our decomposition-based approach, respectively.
Each entity in the tables is an average value of the repeated experiments (100 times).
The numbers in the first column and the first row represent the source and target attractors, respectively.
The results show that our decomposition-based approach outperforms the global approach for any combination of source and target attractors.
It is also obvious that the speedups are highly related to the target attractors.
The speedups with different target attractors vary a lot regarding to the same source attractor.

Table~\ref{tab:networks} gives an overview of the two biological networks and their evaluation results.
For the PC12 cell network, the ranges of the time costs of the global approach and our decomposition-based approach are $16-56$ (ms) and $5-12$ (ms) resp.
The speedups gained by our decomposition-based approach are between $1.375$ and $9.672$.
For the apoptosis network, the ranges of the time costs of the global approach and our decomposition-based approach are $1,472-46,560$ (ms) and $747-994$ (ms) resp.
The speedups gained by our decomposition-based approach are between $1.932$ and $51.504$.
Benefited from the fixpoint computation of strong basin, described in Algorithm~\ref{alg:stb}, both approaches are efficient.
Compared with the global approach, our decomposition-based approach has an evident advantage in terms of efficiency, especially for large networks.

\begin{table*}[!t] 
\centering 
\begin{tabular}{|C{1.5cm}|C{0.5cm}C{0.5cm}|C{0.5cm}C{0.5cm}|C{0.5cm}C{0.5cm}|C{0.5cm}C{0.5cm}|C{0.5cm}C{0.5cm}|C{0.5cm}C{0.5cm}|C{0.5cm}C{0.5cm}|}
 \hline 
\multicolumn{1}{|c|}{\multirow{2}{*}{Attractor}} & \multicolumn{2}{c|}{1} & \multicolumn{2}{c|}{2} & \multicolumn{2}{c|}{3} & \multicolumn{2}{c|}{4} & \multicolumn{2}{c|}{5} & \multicolumn{2}{c|}{6} & \multicolumn{2}{c|}{7} 
 \\ \cline{2-15}
\multicolumn{1}{|c|}{} & \multicolumn{1}{c}{HD} & \multicolumn{1}{c|}{\#D} & \multicolumn{1}{c}{HD} & \multicolumn{1}{c|}{\#D} & \multicolumn{1}{c}{HD} & \multicolumn{1}{c|}{\#D} & \multicolumn{1}{c}{HD} & \multicolumn{1}{c|}{\#D} & \multicolumn{1}{c}{HD} & \multicolumn{1}{c|}{\#D} & \multicolumn{1}{c}{HD} & \multicolumn{1}{c|}{\#D} & \multicolumn{1}{c}{HD} & \multicolumn{1}{c|}{\#D} 
 \\ \hline
$1$ &$-$ &$-$ &$2$ &$1$ &$21$ &$7$ &$22$ &$8$ &$22$ &$8$ &$23$ &$9$ &$7$ &$1$ \\ 
$2$ &$2$ &$1$ &$-$ &$-$ &$23$ &$8$ &$22$ &$7$ &$24$ &$9$ &$23$ &$8$ &$9$ &$1$ \\ 
$3$ &$21$ &$10$ &$23$ &$11$ &$-$ &$-$ &$1$ &$1$ &$1$ &$1$ &$2$ &$2$ &$28$ &$1$ \\ 
$4$ &$22$ &$11$ &$22$ &$10$ &$1$ &$1$ &$-$ &$-$ &$2$ &$2$ &$1$ &$1$ &$29$ &$1$ \\ 
$5$ &$22$ &$10$ &$24$ &$11$ &$1$ &$1$ &$2$ &$2$ &$-$ &$-$ &$1$ &$1$ &$29$ &$1$ \\ 
$6$ &$23$ &$11$ &$23$ &$10$ &$2$ &$2$ &$1$ &$1$ &$1$ &$1$ &$-$ &$-$ &$30$ &$1$ \\ 
$7$ &$7$ &$1$ &$9$ &$3$ &$28$ &$9$ &$29$ &$10$ &$29$ &$10$ &$30$ &$11$ &$-$ &$-$ \\ 
\hline 
\end{tabular}
\caption{The Hamming distance between attractors and the number of driver nodes computed using the global and decomposition-based approaches on the PC12 cell network.}
\label{tab:dr-pc12}
\end{table*}

\begin{table*}[!t] 
\centering 
\begin{tabular}{|C{1.5cm}|C{0.5cm}C{0.5cm}|C{0.5cm}C{0.5cm}|C{0.5cm}C{0.5cm}|C{0.5cm}C{0.5cm}|C{0.5cm}C{0.5cm}|C{0.5cm}C{0.5cm}|C{0.5cm}C{0.5cm}|}
 \hline 
\multicolumn{1}{|c|}{\multirow{2}{*}{Attractor}} & \multicolumn{2}{c|}{1} & \multicolumn{2}{c|}{2} & \multicolumn{2}{c|}{3} & \multicolumn{2}{c|}{4} & \multicolumn{2}{c|}{5} & \multicolumn{2}{c|}{6} & \multicolumn{2}{c|}{7} 
 \\ \cline{2-15}
\multicolumn{1}{|c|}{} & \multicolumn{1}{c}{HD} & \multicolumn{1}{c|}{\#D} & \multicolumn{1}{c}{HD} & \multicolumn{1}{c|}{\#D} & \multicolumn{1}{c}{HD} & \multicolumn{1}{c|}{\#D} & \multicolumn{1}{c}{HD} & \multicolumn{1}{c|}{\#D} & \multicolumn{1}{c}{HD} & \multicolumn{1}{c|}{\#D} & \multicolumn{1}{c}{HD} & \multicolumn{1}{c|}{\#D} & \multicolumn{1}{c}{HD} & \multicolumn{1}{c|}{\#D} 
 \\ \hline
$9$ &$1$ &$1$ &$6$ &$2$ &$3$ &$2$ &$8$ &$3$ &$21$ &$2$ &$26$ &$3$ &$23$ &$3$ \\ 
$10$ &$6$ &$2$ &$1$ &$1$ &$8$ &$3$ &$3$ &$2$ &$26$ &$3$ &$21$ &$2$ &$28$ &$4$  \\ 
$11$ &$3$ &$2$ &$8$ &$3$ &$1$ &$1$ &$6$ &$2$ &$23$ &$3$ &$28$ &$4$ &$21$ &$2$ \\ 
$12$ &$8$ &$3$ &$3$ &$2$ &$6$ &$2$ &$1$ &$1$ &$28$ &$4$ &$23$ &$3$ &$26$ &$3$  \\ 
$13$ &$10$ &$2$ &$15$ &$3$ &$12$ &$3$ &$17$ &$4$ &$12$ &$1$ &$17$ &$2$ &$14$ &$2$  \\ 
$14$ &$15$ &$3$ &$10$ &$2$ &$17$ &$4$ &$12$ &$3$ &$17$ &$2$ &$12$ &$1$ &$19$ &$3$  \\ 
$15$ &$12$ &$3$ &$17$ &$4$ &$10$ &$2$ &$15$ &$3$ &$14$ &$2$ &$19$ &$3$ &$12$ &$1$  \\ 
$16$ &$17$ &$4$ &$12$ &$3$ &$15$ &$3$ &$10$ &$2$ &$19$ &$3$ &$14$ &$2$ &$17$ &$2$  \\ 
\hline 
\end{tabular}
\caption{The Hamming distance between attractors and the number of driver nodes computed using the global and decomposition-based approaches on the apoptosis network.}
\label{tab:dr-apop}
\end{table*}

\begin{table}[!t] 
\centering 
\begin{tabular}{|C{1.3cm}|C{0.6cm}|C{0.6cm}|C{0.6cm}|C{0.6cm}|C{0.6cm}|C{0.6cm}|C{0.6cm}|}
 \hline 
\multicolumn{1}{|c|}{\multirow{2}{*}{Attractor}} & \multicolumn{7}{c|}{Speedups}
 \\ \cline{2-8} 
\multicolumn{1}{|c|}{} & \multicolumn{1}{c|}{~~~~~~$1$~~~~~~} & \multicolumn{1}{c|}{~~~~~~$2$~~~~~~} & \multicolumn{1}{c|}{~~~~~~$3$~~~~~~} & \multicolumn{1}{c|}{~~~~~~$4$~~~~~~} & \multicolumn{1}{c|}{~~~~~~$5$~~~~~~} & \multicolumn{1}{c|}{~~~~~~$6$~~~~~~} & \multicolumn{1}{c|}{~~~~~~$7$~~~~~~} 
 \\ \hline
$1$ &$-$ &$5.08$ &$9.63$ &$6.14$ &$5.52$ &$4.46$ &$1.38$ \\ 
$2$ &$5.57$ &$-$ &$9.38$ &$6.37$ &$5.50$ &$4.42$ &$1.38$ \\ 
$3$ &$5.44$ &$4.26$ &$-$ &$8.11$ &$5.66$ &$3.14$ &$1.93$ \\ 
$4$ &$5.34$ &$4.26$ &$9.67$ &$-$ &$5.36$ &$3.41$ &$2.30$ \\ 
$5$ &$5.20$ &$4.08$ &$8.99$ &$6.04$ &$-$ &$3.70$ &$1.85$ \\ 
$6$ &$5.25$ &$4.28$ &$9.65$ &$5.83$ &$5.74$ &$-$ &$1.91$ \\ 
$7$ &$5.29$ &$4.27$ &$9.66$ &$5.72$ &$5.70$ &$4.33$ &$-$ \\ 
\hline 
\end{tabular}
\caption{Speedups gained by the decomposition-based approach on the PC12 network.}
\label{tab:sp-pc12}
\end{table}

\begin{table}[!t] 
\centering 
\begin{tabular}{|C{1.3cm}|C{0.6cm}|C{0.6cm}|C{0.6cm}|C{0.6cm}|C{0.6cm}|C{0.6cm}|C{0.6cm}|}
 \hline 
\multicolumn{1}{|c|}{\multirow{2}{*}{Attractor}} & \multicolumn{7}{c|}{Speedups}
 \\ \cline{2-8}  
\multicolumn{1}{|c|}{} & \multicolumn{1}{c|}{~~~~~~$1$~~~~~~} & \multicolumn{1}{c|}{~~~~~~$2$~~~~~~} & \multicolumn{1}{c|}{~~~~~~$3$~~~~~~} & \multicolumn{1}{c|}{~~~~~~$4$~~~~~~} & \multicolumn{1}{c|}{~~~~~~$5$~~~~~~} & \multicolumn{1}{c|}{~~~~~~$6$~~~~~~} & \multicolumn{1}{c|}{~~~~~~$7$~~~~~~} 
 \\ \hline
$9$  &$2.06$ &$2.00$ &$2.26$ &$2.22$ &$20.27$ &$17.44$ &$24.88$  \\ 
$10$ &$2.13$ &$1.93$ &$2.37$ &$2.14$ &$22.13$ &$17.58$ &$24.95$  \\ 
$11$ &$2.14$ &$1.94$ &$2.37$ &$2.15$ &$22.22$ &$17.24$ &$24.71$  \\ 
$12$ &$2.17$ &$1.98$ &$2.38$ &$2.18$ &$22.11$ &$17.12$ &$25.31$  \\ 
$13$ &$2.17$ &$1.96$ &$2.40$ &$2.18$ &$22.58$ &$17.65$ &$24.86$  \\ 
$14$ &$2.15$ &$1.95$ &$2.40$ &$2.17$ &$22.30$ &$17.46$ &$25.09$  \\ 
$15$ &$2.14$ &$1.95$ &$2.39$ &$2.17$ &$22.24$ &$18.09$ &$24.85$  \\ 
$16$ &$2.16$ &$1.97$ &$2.42$ &$2.19$ &$22.58$ &$17.53$ &$25.23$  \\ 
\hline 
\end{tabular}
\caption{Speedups gained by the decomposition-based approach on the apoptosis network.}
\label{tab:sp-apop}
\end{table}

\begin{table*}[!t]
\centering
\begin{tabular}{|C{1.5cm}|C{1cm}C{1cm}C{1.5cm}C{2.4cm}C{1.8cm}C{2.4cm}|}
\hline 
\multirow{2}{*}{Networks}  & \multirow{2}{*}{\begin{tabular}[c]{@{}c@{}}\# \\ ~~~nodes~~~\end{tabular}}  & \multirow{2}{*}{\begin{tabular}[c]{@{}c@{}}\# \\ ~~~blocks~~~\end{tabular}} & \multirow{2}{*}{\begin{tabular}[c]{@{}c@{}}\# \\ ~~~attractors~~~\end{tabular}} &Range of & Range of & Range of\\
 & & & & $t_{\it global}$ (ms) & $t_{\it decom}$ (ms) & speedups\\ \hline
PC12      & $32$    & $19$   & $7$  &$16-56$ &$5-12$ & $1.375-9.672$  \\
apoptosis & $97$    & $60$   & $16$ & $1,472-46,560$ & $747-994$ & $1.932-51.504$ \\ \hline
BN-100    & $100$   & $36$   & $9$  & $3,275-4424,750$ & $94-1,141$ & $11.738-35973.577$ \\  
BN-120    & $120$   & $27$   & $4$  & $257,3-14774,2$ & $2,840 - 6,466$ & $39,89-4818,72$ \\  
BN-180    & $180$   & $62$   & $2$  & $*$ & $1,402 - 1,462$ & $*$ \\  
 \hline
\end{tabular}
\caption{An overview of the evaluation results of the two real-life biological networks and the three randomly generated Boolean networks. 
The $*$ means the program fails to return any results within five hours.}
\label{tab:networks}
\end{table*}

\subsection{Case studies on randomly generated networks}
The same procedures are applied to three randomly generated Boolean networks with 100, 120 and 180 nodes.
An overview of the three networks and their evaluation results is given in Table~\ref{tab:networks}.
The BNs with 100, 120 and 180 nodes are labelled as BN-100, BN-120 and BN-180 and they have 9, 4 and 2 single-state attractors, respectively. 
The global approach fails to compute the driver nodes for the BN-180 network and for some cases of the BN-100 and BN-120 networks. 
The corresponding results are denoted as $*$.  
The range of the time costs of the decomposition-based approach for the BN-180 network is $1,402 - 1,462$ (ms).
For the BN-120 network, the ranges of the time costs of the global approach and our decomposition-based approach are $257,3-14774,2$ (ms) and $2,840-6,466$ (ms) resp.

Table~\ref{tab:BN-100} shows the time costs of the global approach and the decomposition-based approach on the BN-100 network.
When the target attractors are 1, 6 and 8, the global approach fails to return any results within five hours.
From Table~\ref{tab:BN-100}, it is clear that the execution time is highly dependent on the target attractor.
Especially for the global approach, it may cost a considerable amount of time when the basin of the target attractor is large.
In terms of the number of driver nodes, the results computed by the two approaches are identical (not shown here).

From experimental results on three randomly generated BNs,
we can conclude that the proposed decomposition-based approach scales well for large networks,
thanks to its `divide and conquer' strategy,
while the global approach fails to compute the results in some cases
due to the fact that it deals with the entire networks at once.

\begin{table*}[!t] 
\centering 
\begin{tabular}{|C{1.5cm}|C{1cm}C{1cm}C{1cm}C{1cm}C{1cm}C{1cm}C{1cm}C{1cm}C{1cm}|}
 \hline 
\multicolumn{1}{|c|}{\multirow{2}{*}{Attractors}} & \multicolumn{9}{c|}{Time (ms)}
 \\ \cline{2-10} 
\multicolumn{1}{|c|}{} & \multicolumn{1}{c}{~~~~~~~~$1$~~~~~~~~} & \multicolumn{1}{c}{~~~~~~~~$2$~~~~~~~~} & \multicolumn{1}{c}{~~~~~~~~$3$~~~~~~~~} & \multicolumn{1}{c}{~~~~~~~~$4$~~~~~~~~} & \multicolumn{1}{c}{~~~~~~~~$5$~~~~~~~~} & \multicolumn{1}{c}{~~~~~~~~$6$~~~~~~~~} & \multicolumn{1}{c}{~~~~~~~~$7$~~~~~~~~} & \multicolumn{1}{c}{~~~~~~~~$8$~~~~~~~~} & \multicolumn{1}{c|}{~~~~~~~~$9$~~~~~~~~} 
 \\ \hline
\multirow{2}{*}{$1$} &$-$ &$20,812$ &$70,816$ &$3,674$ &$8,278$ &$*$ &$4182,340$ &$*$ &$3,935$ \\ 
&$-$ &$595$ &$388$ &$279$ &$259$ &$276$ &$131$ &$114$ &$96$ \\ \hline 
\multirow{2}{*}{$2$} &$*$ &$-$ &$74,023$ &$3,694$ &$7,732$ &$*$ &$4072,920$ &$*$ &$3,808$ \\ 
&$1,136$ &$-$ &$389$ &$281$ &$260$ &$275$ &$123$ &$112$ &$96$ \\ \hline 
\multirow{2}{*}{$3$} &$*$ &$19,956$ &$-$ &$3,706$ &$7,688$ &$*$ &$3750,570$ &$*$ &$3,865$ \\ 
&$1,138$ &$595$ &$-$ &$282$ &$260$ &$275$ &$123$ &$112$ &$98$ \\ \hline 
\multirow{2}{*}{$4$} &$*$ &$19,747$ &$72,511$ &$-$ &$7,721$ &$*$ &$2021,220$ &$*$ &$3,904$ \\ 
&$1,138$ &$595$ &$388$ &$-$ &$260$ &$276$ &$123$ &$112$ &$97$ \\ \hline 
\multirow{2}{*}{$5$} &$*$ &$22,598$ &$29,760$ &$3,275$ &$-$ &$*$ &$4424,750$ &$*$ &$4,249$ \\ 
&$1,137$ &$595$ &$389$ &$279$ &$-$ &$275$ &$123$ &$112$ &$94$ \\ \hline 
\multirow{2}{*}{$6$} &$*$ &$19,744$ &$73,355$ &$3,707$ &$7,750$ &$-$ &$2149,410$ &$*$ &$3,883$ \\ 
&$1,141$ &$595$ &$389$ &$282$ &$259$ &$-$ &$121$ &$111$ &$96$ \\ \hline 
\multirow{2}{*}{$7$} &$*$ &$19,742$ &$72,197$ &$3,706$ &$7,689$ &$*$ &$-$ &$*$ &$3,886$ \\ 
&$1,139$ &$595$ &$390$ &$280$ &$259$ &$274$ &$-$ &$111$ &$97$ \\ \hline 
\multirow{2}{*}{$8$} &$*$ &$19,763$ &$73,115$ &$3,706$ &$7,701$ &$*$ &$2089,800$ &$-$ &$3,842$ \\ 
&$1,139$ &$594$ &$388$ &$281$ &$259$ &$274$ &$124$ &$-$ &$96$ \\ \hline 
\multirow{2}{*}{$9$} &$*$ &$19,719$ &$73,397$ &$3,710$ &$7,460$ &$*$ &$2343,120$ &$*$ &$-$ \\ 
&$1,139$ &$595$ &$389$ &$283$ &$259$ &$274$ &$124$ &$111$ &$-$ \\ \hline 
 \end{tabular}
\caption{Time costs of the global approach and the decomposition-based approach on the BN-100 network.
The $*$ means the program fails to return any results within five hours.}
\label{tab:BN-100}
\end{table*}



\section{Conclusions and Future Work}
\label{sec:conclusion}
In this work, we have described a decomposition-based approach towards the computation of a minimal set
of nodes (variables) to be simultaneously controlled of a BN so as to drive its dynamics from a source
state to a target attractor. Our approach is generic and can be applied based on any algorithm for
computing the strong basin of attraction of an attractor. For certain modular real-life networks, the
approach results in significant increase in efficiency compared with a global approach and its
generality means that the improvement in efficiency can be attained irrespective of the exact algorithm
used for the computation of the strong basins.

We have only scratched the surface of what we believe to be an exciting approach towards the control of
BNs which utilises both its structure and dynamics. We conclude by looking back critically
at our approach, summarising various extensions and discussing future directions.

As mentioned in Section~\ref{sec:intro}, the problem of minimal control is PSPACE-hard and efficient algorithms
are unlikely for the general cases. Yet in retrospect, one might ask what is the inherent characteristic of
our decomposition-based approach that makes it so efficient compared with the global approach for the
real-life networks that we studied.  We put forward a couple of heuristics which we believe explains and
crucially determines the success of our approach. One such heuristic is that the basins of attraction
computed at each step is small compared with the size of the transition system. 
This reduces the state space that needs to be considered in every subsequent step thus improving efficiency.

Another heuristic, which depends on the structure of the network, is that the number of blocks is small
compared with the total number of nodes in the network. Otherwise, the approach has to compute a large
number of local transition systems (as many as the number of blocks) which hampers its efficiency.
However, the number of blocks in the network cannot be too few either. Otherwise, our approach comes
close to the global approach in terms of efficiency. Note that if the entire network is one single giant
block, then the decomposition-based approach is the same as the global approach (given that the same
procedure is used for the computation of the strong basins) and there is no gain in efficiency.
One might thus conjecture that there is an optimal block-to-node ratio, given which, our
decomposition-based approach fares the best.

As discussed at the end of Section~\ref{sec:tsblocks}, in~\cite{MPY17b} the construction of the TS of a
non-elementary block $B$ depends on the transitions of the control nodes of $B$ which can be derived
by projecting the transitions in the attractors of the parent block(s) of $B$ to these control nodes.
By this process of projection, the states of the TS of $B$ had smaller dimension (equal to $|B|$) as
compared with our current approach where the states of $B$ have dimension equal to $|\ac(B)|$. This,
in effect, can speed up the decomposition-based approach. Unfortunately, it turns out that such a
projection does not work when we require to preserve the basins of the attractors across the blocks.
Projection results in loss of information, without which it is not possible to derive the global
basin of an attractor of the entire BN in terms of the cross of the local basins.
However, it can be shown that if we do generate the transition system of a non-elementary block $B$ by
projecting the basins of attractions of the parent blocks to the control nodes of $B$, the cross of the
local basins is a subset of the corresponding global basin of the attractor of the entire network.
Thus, if we are ready to sacrifice accuracy for efficiency, such a projection-based technique might be
faster for certain networks while not exactly giving the minimal nodes to control but a good-enough
approximation of it. We would like to study the gain in efficiency in our approach by applying the above
technique. 

One way to reduce the number of `small' blocks (which, as discussed, might degrade efficiency) might
be to combine multiple basic blocks into larger blocks. While constructing the local transition systems,
such merged blocks are treated as single basic blocks and their dynamics, attractors and basins are
computed in one-go. We believe there are many real-life networks which might benefit from such a process
of merging before applying our decomposition-based approach for control. This is another line of work that
we are pursuing at the moment.
As mentioned in the related work, the control approaches based on computation of the feedback vertex
set~\cite{ABGD13,BAGD13,ZYA17} and the stable motifs~\cite{ZA15} are promising approximate control algorithms for nonlinear dynamical networks.
We would like to compare our approaches with these two in terms of efficiency and the number of driver nodes. 
Finally, we plan to extend our decomposition-based approach to the control of probabilistic Boolean
networks~\cite{SD10,TMPTSS13}.


\begin{acks}
S. Paul and C. Su were supported by the research project SEC-PBN funded by the University of Luxembourg.
This work was also partially supported by the ANR-FNR project AlgoReCell ({\sf INTER/ANR/15/11191283}).
\end{acks}

\balance
\bibliographystyle{ACM-Reference-Format}
\bibliography{control}


\begin{thebibliography}{24}


\ifx \showCODEN    \undefined \def \showCODEN     #1{\unskip}     \fi
\ifx \showDOI      \undefined \def \showDOI       #1{#1}\fi
\ifx \showISBNx    \undefined \def \showISBNx     #1{\unskip}     \fi
\ifx \showISBNxiii \undefined \def \showISBNxiii  #1{\unskip}     \fi
\ifx \showISSN     \undefined \def \showISSN      #1{\unskip}     \fi
\ifx \showLCCN     \undefined \def \showLCCN      #1{\unskip}     \fi
\ifx \shownote     \undefined \def \shownote      #1{#1}          \fi
\ifx \showarticletitle \undefined \def \showarticletitle #1{#1}   \fi
\ifx \showURL      \undefined \def \showURL       {\relax}        \fi
\providecommand\bibfield[2]{#2}
\providecommand\bibinfo[2]{#2}
\providecommand\natexlab[1]{#1}
\providecommand\showeprint[2][]{arXiv:#2}

\bibitem[\protect\citeauthoryear{Czeizler, Gratie, Chiu, Kanhaiya, and
  Petre}{Czeizler et~al\mbox{.}}{2016}]%
        {CGCK16}
\bibfield{author}{\bibinfo{person}{E. Czeizler}, \bibinfo{person}{C. Gratie},
  \bibinfo{person}{W.~K. Chiu}, \bibinfo{person}{K. Kanhaiya}, {and}
  \bibinfo{person}{I. Petre}.} \bibinfo{year}{2016}\natexlab{}.
\newblock \showarticletitle{Target Controllability of Linear Networks}. In
  \bibinfo{booktitle}{\emph{Proc.\ 14th International Conference on
  Computational Methods in Systems Biology}} \emph{(\bibinfo{series}{LNCS})},
  Vol.~\bibinfo{volume}{9859}. \bibinfo{publisher}{Springer},
  \bibinfo{pages}{67--81}.
\newblock


\bibitem[\protect\citeauthoryear{Fiedler, Mochizuki, Kurosawa, and
  Saito}{Fiedler et~al\mbox{.}}{2013}]%
        {BAGD13}
\bibfield{author}{\bibinfo{person}{B. Fiedler}, \bibinfo{person}{A. Mochizuki},
  \bibinfo{person}{G. Kurosawa}, {and} \bibinfo{person}{D. Saito}.}
  \bibinfo{year}{2013}\natexlab{}.
\newblock \showarticletitle{Dynamics and control at feedback vertex sets. I:
  Informative and determining nodes in regulatory networks}.
\newblock \bibinfo{journal}{\emph{Journal of Dynamics and Differential
  Equations}} \bibinfo{volume}{25}, \bibinfo{number}{3} (\bibinfo{year}{2013}),
  \bibinfo{pages}{563--604}.
\newblock


\bibitem[\protect\citeauthoryear{Gao, Liu, D'Souza, and Barab\'asi}{Gao
  et~al\mbox{.}}{2014}]%
        {GLDB14}
\bibfield{author}{\bibinfo{person}{J. Gao}, \bibinfo{person}{Y.-Y. Liu},
  \bibinfo{person}{R.~M. D'Souza}, {and} \bibinfo{person}{A.-L. Barab\'asi}.}
  \bibinfo{year}{2014}\natexlab{}.
\newblock \showarticletitle{Target control of complex networks}.
\newblock \bibinfo{journal}{\emph{Nature Communications}}  \bibinfo{volume}{5}
  (\bibinfo{year}{2014}), \bibinfo{pages}{5415}.
\newblock


\bibitem[\protect\citeauthoryear{Gates and Rocha}{Gates and Rocha}{2016}]%
        {GR16}
\bibfield{author}{\bibinfo{person}{A.~J. Gates} {and} \bibinfo{person}{L.~M.
  Rocha}.} \bibinfo{year}{2016}\natexlab{}.
\newblock \showarticletitle{Control of complex networks requires both structure
  and dynamics}.
\newblock \bibinfo{journal}{\emph{Scientific Reports}} \bibinfo{volume}{6},
  \bibinfo{number}{24456} (\bibinfo{year}{2016}).
\newblock


\bibitem[\protect\citeauthoryear{Graf and Enver}{Graf and Enver}{2009}]%
        {Nature09}
\bibfield{author}{\bibinfo{person}{T. Graf} {and} \bibinfo{person}{T. Enver}.}
  \bibinfo{year}{2009}\natexlab{}.
\newblock \showarticletitle{Forcing cells to change lineages}.
\newblock \bibinfo{journal}{\emph{Nature}} \bibinfo{volume}{462},
  \bibinfo{number}{7273} (\bibinfo{year}{2009}), \bibinfo{pages}{587--594}.
\newblock


\bibitem[\protect\citeauthoryear{Huang}{Huang}{2001}]%
        {HS01}
\bibfield{author}{\bibinfo{person}{S. Huang}.} \bibinfo{year}{2001}\natexlab{}.
\newblock \showarticletitle{Genomics, complexity and drug discovery: insights
  from {B}oolean network models of cellular regulation}.
\newblock \bibinfo{journal}{\emph{Pharmacogenomics}} \bibinfo{volume}{2},
  \bibinfo{number}{3} (\bibinfo{year}{2001}), \bibinfo{pages}{203--222}.
\newblock


\bibitem[\protect\citeauthoryear{Kauffman}{Kauffman}{1969}]%
        {KS69}
\bibfield{author}{\bibinfo{person}{S.~A. Kauffman}.}
  \bibinfo{year}{1969}\natexlab{}.
\newblock \showarticletitle{Homeostasis and differentiation in random genetic
  control networks}.
\newblock \bibinfo{journal}{\emph{Nature}}  \bibinfo{volume}{224}
  (\bibinfo{year}{1969}), \bibinfo{pages}{177--178}.
\newblock


\bibitem[\protect\citeauthoryear{Lai}{Lai}{2014}]%
        {Lai14}
\bibfield{author}{\bibinfo{person}{Y.-C. Lai}.}
  \bibinfo{year}{2014}\natexlab{}.
\newblock \showarticletitle{Controlling complex, non-linear dynamical
  networks}.
\newblock \bibinfo{journal}{\emph{National Science Review}}
  \bibinfo{volume}{1}, \bibinfo{number}{3} (\bibinfo{year}{2014}),
  \bibinfo{pages}{339--341}.
\newblock


\bibitem[\protect\citeauthoryear{Liu, Slotine, and Barab\'asi}{Liu
  et~al\mbox{.}}{2011}]%
        {LSB11}
\bibfield{author}{\bibinfo{person}{Y.-Y. Liu}, \bibinfo{person}{J.-J. Slotine},
  {and} \bibinfo{person}{A.-L. Barab\'asi}.} \bibinfo{year}{2011}\natexlab{}.
\newblock \showarticletitle{Controllability of complex networks}.
\newblock \bibinfo{journal}{\emph{Nature}}  \bibinfo{volume}{473}
  (\bibinfo{year}{2011}), \bibinfo{pages}{167--€"--173}.
\newblock


\bibitem[\protect\citeauthoryear{Lomuscio, Qu, and Raimondi}{Lomuscio
  et~al\mbox{.}}{2017}]%
        {MCMAS}
\bibfield{author}{\bibinfo{person}{A. Lomuscio}, \bibinfo{person}{H. Qu}, {and}
  \bibinfo{person}{F. Raimondi}.} \bibinfo{year}{2017}\natexlab{}.
\newblock \showarticletitle{{MCMAS}: An open-source model checker for the
  verification of multi-agent systems}.
\newblock \bibinfo{journal}{\emph{International Journal on Software Tools for
  Technology Transfer}} \bibinfo{volume}{19}, \bibinfo{number}{1}
  (\bibinfo{year}{2017}), \bibinfo{pages}{9--30}.
\newblock


\bibitem[\protect\citeauthoryear{Mandon, Haar, and Paulev{\'e}}{Mandon
  et~al\mbox{.}}{2016}]%
        {MHP16}
\bibfield{author}{\bibinfo{person}{H. Mandon}, \bibinfo{person}{S. Haar}, {and}
  \bibinfo{person}{L. Paulev{\'e}}.} \bibinfo{year}{2016}\natexlab{}.
\newblock \showarticletitle{Relationship between the Reprogramming Determinants
  of Boolean Networks and their Interaction Graph}. In
  \bibinfo{booktitle}{\emph{Proc.\ 5th International Workshop on Hybrid Systems
  Biology}} \emph{(\bibinfo{series}{LNCS})}, Vol.~\bibinfo{volume}{9957}.
  \bibinfo{publisher}{Springer}, \bibinfo{pages}{113--127}.
\newblock


\bibitem[\protect\citeauthoryear{Mandon, Haar, and Paulev{\'e}}{Mandon
  et~al\mbox{.}}{2017}]%
        {MHP17}
\bibfield{author}{\bibinfo{person}{H. Mandon}, \bibinfo{person}{S. Haar}, {and}
  \bibinfo{person}{L. Paulev{\'e}}.} \bibinfo{year}{2017}\natexlab{}.
\newblock \showarticletitle{Temporal Reprogramming of Boolean Networks}. In
  \bibinfo{booktitle}{\emph{Proc.\ 15th International Conference on
  Computational Methods in Systems Biology}} \emph{(\bibinfo{series}{LNCS})},
  Vol.~\bibinfo{volume}{10545}. \bibinfo{publisher}{Springer},
  \bibinfo{pages}{179--195}.
\newblock


\bibitem[\protect\citeauthoryear{Marques-Pita and Rocha}{Marques-Pita and
  Rocha}{2013}]%
        {MPR13}
\bibfield{author}{\bibinfo{person}{M. Marques-Pita} {and} \bibinfo{person}{L.M.
  Rocha}.} \bibinfo{year}{2013}\natexlab{}.
\newblock \showarticletitle{Canalization and control in automata networks: body
  segmentation in Drosophila melanogaster}.
\newblock \bibinfo{journal}{\emph{PLoS One}} \bibinfo{volume}{8},
  \bibinfo{number}{3} (\bibinfo{year}{2013}).
\newblock
\newblock
\shownote{e55946.}


\bibitem[\protect\citeauthoryear{Mizera, Pang, Qu, and Yuan}{Mizera
  et~al\mbox{.}}{2018}]%
        {MPY17b}
\bibfield{author}{\bibinfo{person}{A. Mizera}, \bibinfo{person}{J. Pang},
  \bibinfo{person}{H. Qu}, {and} \bibinfo{person}{Q. Yuan}.}
  \bibinfo{year}{2018}\natexlab{}.
\newblock \showarticletitle{Taming Asynchrony for Attractor Detection in Large
  {B}oolean Networks}.
\newblock \bibinfo{journal}{\emph{IEEE/ACM Transactions on Computational
  Biology and Bioinformatics (Special issue of APBC'18)}}
  (\bibinfo{year}{2018}).
\newblock


\bibitem[\protect\citeauthoryear{Mizera, Pang, and Yuan}{Mizera
  et~al\mbox{.}}{2016}]%
        {MPY16b}
\bibfield{author}{\bibinfo{person}{A. Mizera}, \bibinfo{person}{J. Pang}, {and}
  \bibinfo{person}{Q. Yuan}.} \bibinfo{year}{2016}\natexlab{}.
\newblock \showarticletitle{{ASSA-PBN} 2.0: {A} software tool for probabilistic
  {B}oolean networks}. In \bibinfo{booktitle}{\emph{Proc.\ 14th International
  Conference on Computational Methods in Systems Biology}}
  \emph{(\bibinfo{series}{LNCS})}, Vol.~\bibinfo{volume}{9859}.
  \bibinfo{publisher}{Springer}, \bibinfo{pages}{309--315}.
\newblock


\bibitem[\protect\citeauthoryear{Mochizuki, Fiedler, Kurosawa, and
  Saito}{Mochizuki et~al\mbox{.}}{2013}]%
        {ABGD13}
\bibfield{author}{\bibinfo{person}{A. Mochizuki}, \bibinfo{person}{B. Fiedler},
  \bibinfo{person}{G. Kurosawa}, {and} \bibinfo{person}{D. Saito}.}
  \bibinfo{year}{2013}\natexlab{}.
\newblock \showarticletitle{Dynamics and control at feedback vertex sets. {II}:
  A faithful monitor to determine the diversity of molecular activities in
  regulatory networks}.
\newblock \bibinfo{journal}{\emph{J. Theor. Biol.}}  \bibinfo{volume}{335}
  (\bibinfo{year}{2013}), \bibinfo{pages}{130--146}.
\newblock


\bibitem[\protect\citeauthoryear{nudo and Albert}{nudo and Albert}{2015}]%
        {ZA15}
\bibfield{author}{\bibinfo{person}{J.~G. T.~Za\ nudo} {and} \bibinfo{person}{R.
  Albert}.} \bibinfo{year}{2015}\natexlab{}.
\newblock \showarticletitle{Cell fate reprogramming by control of intracellular
  network dynamics}.
\newblock \bibinfo{journal}{\emph{PLoS Computational Biology}}
  \bibinfo{volume}{11}, \bibinfo{number}{4} (\bibinfo{year}{2015}),
  \bibinfo{pages}{e1004193}.
\newblock


\bibitem[\protect\citeauthoryear{Offermann, Knauer, Singh,
  Fern{\'a}ndez-Cach{\'o}n, Klose, Kowar, Busch, and Boerries}{Offermann
  et~al\mbox{.}}{2016}]%
        {OKB16}
\bibfield{author}{\bibinfo{person}{B. Offermann}, \bibinfo{person}{S. Knauer},
  \bibinfo{person}{A. Singh}, \bibinfo{person}{M.~L. Fern{\'a}ndez-Cach{\'o}n},
  \bibinfo{person}{M. Klose}, \bibinfo{person}{S. Kowar}, \bibinfo{person}{H.
  Busch}, {and} \bibinfo{person}{M. Boerries}.}
  \bibinfo{year}{2016}\natexlab{}.
\newblock \showarticletitle{{Boolean} modeling reveals the necessity of
  transcriptional regulation for bistability in {PC12} cell differentiation}.
\newblock \bibinfo{journal}{\emph{F. Genetics}}  \bibinfo{volume}{7}
  (\bibinfo{year}{2016}), \bibinfo{pages}{44}.
\newblock


\bibitem[\protect\citeauthoryear{Schlatter, Schmich, Vizcarra, Scheurich,
  Sauter, Borner, Ederer, Merfort, and Sawodny}{Schlatter
  et~al\mbox{.}}{2009}]%
        {SSVSSBEMS09}
\bibfield{author}{\bibinfo{person}{R. Schlatter}, \bibinfo{person}{K. Schmich},
  \bibinfo{person}{I.~A. Vizcarra}, \bibinfo{person}{P. Scheurich},
  \bibinfo{person}{T. Sauter}, \bibinfo{person}{C. Borner}, \bibinfo{person}{M.
  Ederer}, \bibinfo{person}{I. Merfort}, {and} \bibinfo{person}{O. Sawodny}.}
  \bibinfo{year}{2009}\natexlab{}.
\newblock \showarticletitle{{ON/OFF} and Beyond -~A Boolean Model of
  Apoptosis}.
\newblock \bibinfo{journal}{\emph{PLOS Computational Biology}}
  \bibinfo{volume}{5}, \bibinfo{number}{12} (\bibinfo{year}{2009}),
  \bibinfo{pages}{e1000595}.
\newblock


\bibitem[\protect\citeauthoryear{Shmulevich and Dougherty}{Shmulevich and
  Dougherty}{2010}]%
        {SD10}
\bibfield{author}{\bibinfo{person}{I. Shmulevich} {and} \bibinfo{person}{E.~R.
  Dougherty}.} \bibinfo{year}{2010}\natexlab{}.
\newblock \bibinfo{booktitle}{\emph{Probabilistic Boolean Networks: The
  Modeling and Control of Gene Regulatory Networks}}.
\newblock \bibinfo{publisher}{SIAM Press}.
\newblock


\bibitem[\protect\citeauthoryear{Sol and Buckley}{Sol and Buckley}{2014}]%
        {SC14}
\bibfield{author}{\bibinfo{person}{A.~del Sol} {and} \bibinfo{person}{N.J.
  Buckley}.} \bibinfo{year}{2014}\natexlab{}.
\newblock \showarticletitle{Concise review: A population shift view of cellular
  reprogramming}.
\newblock \bibinfo{journal}{\emph{STEM CELLS}} \bibinfo{volume}{32},
  \bibinfo{number}{6} (\bibinfo{year}{2014}), \bibinfo{pages}{1367--1372}.
\newblock


\bibitem[\protect\citeauthoryear{Trairatphisan, Mizera, Pang, Tantar,
  Schneider, and Sauter}{Trairatphisan et~al\mbox{.}}{2013}]%
        {TMPTSS13}
\bibfield{author}{\bibinfo{person}{P. Trairatphisan}, \bibinfo{person}{A.
  Mizera}, \bibinfo{person}{J. Pang}, \bibinfo{person}{A.-A. Tantar},
  \bibinfo{person}{J. Schneider}, {and} \bibinfo{person}{T. Sauter}.}
  \bibinfo{year}{2013}\natexlab{}.
\newblock \showarticletitle{Recent development and biomedical applications of
  probabilistic {B}oolean networks}.
\newblock \bibinfo{journal}{\emph{Cell Communication and Signaling}}
  \bibinfo{volume}{11} (\bibinfo{year}{2013}), \bibinfo{pages}{46}.
\newblock


\bibitem[\protect\citeauthoryear{Wang, Su, Huang, Wang, Wang, Grebogi, and
  Lai}{Wang et~al\mbox{.}}{2016}]%
        {L16}
\bibfield{author}{\bibinfo{person}{L.-Z. Wang}, \bibinfo{person}{R.-Q. Su},
  \bibinfo{person}{Z.-G. Huang}, \bibinfo{person}{X. Wang},
  \bibinfo{person}{W.-X. Wang}, \bibinfo{person}{C. Grebogi}, {and}
  \bibinfo{person}{Y.-C. Lai}.} \bibinfo{year}{2016}\natexlab{}.
\newblock \showarticletitle{A geometrical approach to control and
  controllability of nonlinear dynamical networks}.
\newblock \bibinfo{journal}{\emph{Nature Communications}}  \bibinfo{volume}{7}
  (\bibinfo{year}{2016}), \bibinfo{pages}{11323}.
\newblock


\bibitem[\protect\citeauthoryear{Za{\~n}udo, Yang, and Albert}{Za{\~n}udo
  et~al\mbox{.}}{2017}]%
        {ZYA17}
\bibfield{author}{\bibinfo{person}{J.~G.~T. Za{\~n}udo}, \bibinfo{person}{G.
  Yang}, {and} \bibinfo{person}{R. Albert}.} \bibinfo{year}{2017}\natexlab{}.
\newblock \showarticletitle{Structure-based control of complex networks with
  nonlinear dynamics}.
\newblock \bibinfo{journal}{\emph{Proceedings of the National Academy of
  Sciences}} \bibinfo{volume}{114}, \bibinfo{number}{28}
  (\bibinfo{year}{2017}), \bibinfo{pages}{7234--7239}.
\newblock


\end{thebibliography}

\clearpage
\appendix


\section{Detailed Proofs}\label{appendix:proofs}
\subsection{Correctness of Algorithm~\ref{alg:stb}}
Define an operator $F$ on $\St$ as follows. For any subset $\T$ of state:
$$F(\T) = \T\setminus (\pre(\post(\T)\setminus\T)\cap\T)$$

It is easy to see that $F$ is monotonically decreasing and hence its greatest fixed point exists. We want to show that for any attractor $A$ of $\ts$, $F^\infty(\bas^W(A))=\bas^S(A)$. That is, to compute the strong basin of $A$ once can start with its weak basin and apply the operator $F$ repeatedly till a fixed point is reached which gives its strong basin. The operation has to be repeated $m$ times where $m$ is the index of $F^\infty(\bas^W(A))$. Note that this would immediately prove the correctness of Algorithm~\ref{alg:stb} since this operation corresponds to the iterative update operation in Algorithm~\ref{alg:stb}, line~\ref{st:fp}. We do so by proving the following lemmas.

\begin{lemma}\label{lem:direct}
  For any state $\state\in \St$, if $\state\notin \bas^S(A)$ then $\state\notin F^\infty(\bas^W(A))$.
\end{lemma}

\begin{proof}
  Suppose for some $\state\in \St$, $\state\notin\bas^S(A)$. Then either (i) there is no path from $\state$ to $A$ or (ii) there is a path from $\state$ to another attractor $A'\neq A$ of $\ts$. If (i) holds then $\state\notin\bas^W(A)$ either and hence $\state\notin F^\infty(\bas^W(A))$. So suppose (ii) holds and there is a path from $\state$ to another attractor $A'\neq A$. Consider the shortest such path $\state_0\rightarrow\state_1\rightarrow\ldots\rightarrow \state_n$, where $\state_0=\state$ and $\state_n\in A'$ and let $\state_i \rightarrow \state_{(i+1)},\ 0\leq i<n$ be the first transition along this path that moves out of $\bas^W(A)$. That is, $\state_i\in \bas^W(A)$ but $\state_{(i+1)}\notin\bas^W(A)$. We claim that $\state\notin F^j(\bas^W(A))$ for all $j\geq (i+1)$. That is, $\state$ is removed in the $(i+1)$th step in the inductive construction of $F^\infty(\bas^W(A))$. We prove this by induction on $i$.

  Suppose $i=0$. Then there is already a transition from $\state$ out of $\bas^W(A)$ and hence $\state\in (\pre(\post(\bas^W(A))\setminus A)\cap \bas^W(A))$. Thus $\state\notin F(\bas^W(A))$. Next, suppose $i> 0$ and the premise holds for all $j: 0\leq j < i$. Then by induction hypothesis we have $\state_1\notin F^i(\bas^W(A))$. Hence $\state\in (\pre(\post(F^i(\bas^W(A)))\setminus F^i(\bas^W(A)))\cap F^i(\bas^W(A)))$. and will be removed in the $(i+1)$th step of the inductive construction.
\end{proof}

For the converse direction, first, we easily observe from the definition of weak and strong basins that:

\begin{lemma}\label{lem:obs}
  Let $A$ be an attractor of $\ts$. Then
  \begin{itemize}
    \item $\bas^S(A)\subseteq \bas^W(A)$,
    \item for any state $\state\in \St$, $\state\in\bas^S(A)$ iff, for all transitions $\state\rightarrow\state'$, we have $\state'\in\bas^S(A)$.
  \end{itemize}
\end{lemma}

We thus have

\begin{lemma}\label{lem:converse}
  For any state $\state\in \St$, if $\state\notin F^\infty(\bas^W(A))$ then $\state\notin \bas^S(A)$.
\end{lemma}

\begin{proof}
  For some state $\state\in\St$, if $\state\notin F^\infty(\bas^W(A))$ then either $\state\notin \bas^W(A)$, in which case $\state\notin\bas^S(A)$ [by Lemma \ref{lem:obs}] or $\state\in\bas^W(A)$ but gets removed from $F^\infty(\bas^W(A))$ at the $i$th step of the inductive construction for some $i\geq 1$. We do an induction on $i$ to show that in that case $\state\notin\bas^S(A)$. Suppose $i=1$. Then by definition $\state\in (\pre(\post(\bas^W(A))\setminus\bas^W(A))\cap \bas^W(A))$ which means there is a transition from $\state$ to some $\state'\notin\bas^W(A)$. Thus $\state\notin \bas^S(A)$ [by Lemma \ref{lem:obs}]. Next suppose $i>1$ and the premise holds for  all $j: 1\leq j< i$. Then, $\state\in (\pre(\post(F^{(i-1)}(\bas^W(A)))\setminus F^{(i-1)}(\bas^W(A)))\cap F^{(i-1)}(\bas^W(A)))$. This means there is a state $\state'\in F^{(i-1)}(\bas^W(A))$ such that there is a transition from $\state$ to $\state'$. But since by induction hypothesis $\state'\notin \bas^S(A)$ we must have that $\state\notin\bas^S(A)$ [by Lemma \ref{lem:obs}].
\end{proof}

Combining Lemmas~\ref{lem:direct} and \ref{lem:converse} we have

\begin{theorem}[Correctness of Algorithm \ref{alg:stb}]\label{thm:strongweak}
  For any attractor $A$ of $\ts$ we have $$\bas^S(A) =  F^\infty(\bas^W(A))$$
\end{theorem}

\subsection{Proofs of Theorem~\ref{thm:attpres} and Theorem~\ref{thm:bassubset}}
Let us start with the case where our given Boolean network $\bn$ has two basic blocks $B_1$ and $B_2$. We shall later generalise the results to the case where $\bn$ has more than two basic blocks by inductive arguments.

Note that either one or both the blocks $B_1$ and $B_2$ are elementary. If only one of the blocks is elementary, we shall without loss in generality, assume that it is $B_1$. Let $\ts,\ts_1$ and $\ts_2$ be the transition systems of $\bn,B_1$ and $B_2$ respectively where, if $B_2$ is non elementary, we shall assume that $\ts_2$ is transition system of $B_2$ generated by the basin of an attractor $A_1$ of $\ts_1$.

The states of a transition system will be denoted by $\state$ or $\tstate$ with appropriate subscripts and/or superscripts. For any state $\state\in \ts$ (resp. $\tstate\in \ts$), we shall denote $\state|_{B_1}$ (resp. $\tstate|_{B_1}$) by $\state_1$ (resp. $\tstate_1$) and $\state|_{B_2}$ (resp. $\tstate|_{B_2}$) by $\state_2$ (resp. $\tstate_2$). Similarly, for a set of states $T$ of $\ts$, $T_1$ and $T_2$ will denote the set of projections of the states in $T$ to $B_1$ and $B_2$ respectively.

Let $B_1^- = B_1\setminus(B_1\cap B_2)$ and $B_2^-= B_2\setminus(B_1\cap B_2)$. We shall denote any transition $\state \longrightarrow \state'$ in $\ts$ by $\state\move{B}\state'$ if the variable whose value changes in the transition is in the set $B$.

\begin{lemma}\label{lem:elemtofull}
  For an elementary block $B_i$ of $\bn$ and for every $\state_i,\state'_i$ of $\ts_i$, if there is a path from $\state_i$ to $\state'_i$ in $\ts_i$, then there is a path from $\state$ to $\state'$ in $\ts$ such that $\state|_{B_i}=\state_i, \state'|_{B_i} = \state'_i$ and $\state|_{B_j^-} = \state'|_{B_j^-},\ j\neq i$.
\end{lemma}

\begin{proof}
  Let $B_i$ be elementary and suppose $\state_i^0 \move{B_1} \state_i^1 \move{B_1} \ldots \move{B_1} \state_i^m$, where $\state_i^0 = \state_i$ and $\state_i^m = \state'_i$, be a path from $\state_i$ to $\state'_i$ in $\ts_i$. Let $\state|_{B_j^-} = \state'|_{B_j^-}=\state^-_j$. It is clear that $(\state_i^0\cross\state^-_j) \move{B_1} (\state_i^1\cross\state^-_j) \move{B_1} \ldots \move{B_1} (\state_i^m\cross\state^-_j)$ is a path from $\state$ to $\state'$ in $\ts$ where $\state=(\state_i^0\cross\state^-_j), \state'=(\state_i^m\cross\state^-_j)$ and $\state$ and $\state'$ have the required properties. Indeed, since $B_i$ is elementary and values of the nodes in $B_j^-$ are not modified along the path.
\end{proof}

\begin{lemma}\label{lem:fulltoelem} For every $\state,\state'$ of $\ts$ if there is a path from $\state$ to $\state'$ in $\ts$ then there is a path from $\state_i$ to $\state'_i$ in $\ts_i$ for every elementary block $B_i$.
\end{lemma}

\begin{proof}
  Suppose $\path=\state^0\rightarrow \state^1\rightarrow \ldots \state^m$, where $\state^0=\state$ and $\state^m=\state'$ be a path from $\state$ to $\state'$ in $\ts$. Let $B_i$ be an elementary block of $\bn$. We inductively construct a path $\path_i$ from $\state_i$ to $\state'_i$ in $\ts_i$ using $\path$. $\path_i^j,\ 0\leq j<m$, will denote the prefix of $\path_i$ constructed in the $j$th step of the induction. Initially $\path_i^0=\state^0_i$. Suppose $\path_i^j$ has been already constructed and consider the next transition $\state^j\rightarrow \state^{j+1}$ in $\path$. If this transition is labeled with $B_i$ then we let $\path_i^{j+1} = \path_i^j\move{B_i}\state^{j+1}_i$. Otherwise if this transition is labeled with $B^-_j, j\neq i$, then we let $\path_i^{j+1} = \path_i^j$. Since by induction hypothesis $\path_i^j$ is a path in $\ts_i$ and we add to this a transition from $\path$ only if a node of the elementary block $B_i$ is modified in this transition, such a transition exists in $\ts_i$. Hence, $\path_i^{j+1}$ is also a path in $\ts_i$. Continuing in this manner, we shall have a path from $\state_i$ to $\state'_i$ in $\ts_i$ at the last step when $j+1=m$.
\end{proof}

\begin{lemma}\label{lem:disjoint}
  Suppose $B_1$ and $B_2$ are both elementary blocks and $B_1\cap B_2=\emptyset$. Then for every $\state,\state'\in \ts$, there is a path from $\state$ to $\state'$ in $\ts$ if and only if there is a path from $\state_i$ to $\state'_i$ in every $\ts_i$.
\end{lemma}

\begin{proof}
  Follows directly from Lemma~\ref{lem:elemtofull} and Lemma~\ref{lem:fulltoelem}.
\end{proof}



\begin{lemma}\label{lem:attrdisj} Let $B_1$ and $B_2$ be two elementary blocks of $\bn$, $B_1\cap B_2=\emptyset$. Then we have that $A$ is an attractor of of $\ts$ if and only if there are attractors $A_1$ and $A_2$ of $\ts_1$ and $\ts_2$ resp. such that $A=A_1\cross A_2$.
\end{lemma}

\begin{proof}
  Follows directly from Lemma~\ref{lem:disjoint}.
\end{proof}

\begin{lemma}\label{lem:attrjoint} Let $\bn$ have two blocks $B_1$ and $B_2$ where $B_2$ is non-elementary, $B_1$ is elementary and is the parent of $B_2$. Then we have $A$ is an attractor of $\ts$ if and only if $A_1$ is an attractor of $\ts_1$ and $A$ is also an attractor of $\ts_2$ where $\ts_2$ is realized by $\bas(A_1)$.
\end{lemma}

\begin{proof}
  Suppose $A$ is an attractor of $\ts$ and for contradiction suppose $A_1$ is not an attractor of $\ts_1$. Then either there exist $\state,\state'\in A$ such that there is no path from $\state_1$ to $\state'_1$ in $\ts_1$. But that is not possible by Lemma~\ref{lem:fulltoelem}. Or there exist $\state_1\in A_1$ and $\state'_1\notin A_1$ such that there is a transition from $\state_1$ to $\state'_1$. But then by Lemma~\ref{lem:elemtofull}, there is a transition from $\state\in A$ to $\state'\notin A$ in $\ts$ where $\state|_{B_1}=\state_1$ and $\state'|_{B_2}=\state'_2$. This contradicts the assumption that $A$ is an attractor of $A$. Next suppose $A$ is not an attractor of $\ts_2$. Then there is a transition in $\ts_2$ from $\state\in A$ to $\state'\notin A$. But we have, by the construction of $\ts_2$ (Definition~\ref{def:tsblocks}), that this is also a transition in $\ts$ which again contradicts the assumption that $A$ is an attractor of $\ts$.

  For the converse direction, suppose for contradiction that $A$ is an attractor of $\ts_2$ and $A_1$ is an attractor of $\ts_1$ but $A$ is not an attractor of $\ts$. We must then have that there is a transition in $\ts$ from $\state\in A$ to $\state'\notin A$. If this transition is labelled with $B_1$ then we must have, by Lemma~\ref{lem:fulltoelem}, that there is a transition in $\ts_1$ from $\state_1$ to $\state'_1$. But since $\state'_1\notin A_1$ this contradicts the assumption that $A_1$ is an attractor of $\ts_1$. Next, suppose that this transition is labelled with $B_2^{-}$. We must then have that $\state_1 = \state'_1\in A_1$. Hence, by the construction of $\ts_2$ (Definition~\ref{def:tsblocks}) it must be the case that $\state'\in \ts_2$ and this transition from $\state$ to $\state'$ is also present in $\ts_2$. But this contradicts the assumption that $A$ is an attractor of $\ts_2$.
\end{proof}

Now suppose $\bn$ has $k$ blocks that are topologically sorted as $\{B_1,B_2,\ldots, B_k\}$. Note that for every $i$ such that $1\leq i\leq k$, ($\bigcup_{j\leq i}B_j$) is an elementary block of $\bn$ and we denote its transition system by $\tshat_i$.

{
\renewcommand{\thetheorem}{\ref{thm:attpres}}
\begin{theorem}[preservation of attractors] Suppose $\bn$ has $k$ basic blocks that are topologically sorted as $\{B_1,B_2,\ldots, B_k\}$. Suppose for every attractor $A$ of $\ts$ and for every $i:1\leq i<k$, if $B_{i+1}$ is non-elementary then $\ts_{i+1}$ is realized by $\bas(\cross_{j\in I}A_j)$, its basin w.r.t. the TS for $(\bigcup_{j\in I}B_j)$, where $I$ is the set of indices of the basic blocks in $\ac(B_{i+1})^-$. We then have, for every $i:1\leq i<k$, $A_{i+1}$ is an attractor of $\ts_{i+1}$,  $(\cross_{j\in I}A_j\cross A_{i+1})$ is an attractor of the TS for the elementary block $(\bigcup_{j\in I}B_j\cup B_{i+1})$, $(\cross_{j=1}^{i+1}A_j)$ is an attractor of $\tshat_{i+1}$ and $A$ is an attractor of $\ts_k$.
\end{theorem}
\addtocounter{theorem}{-1}
}

\begin{proof}
  The proof is by induction on $i$. The base case is when $i=2$ and $\bn$ has two blocks $B_1$ and $B_2$. If $B_1$ and $B_2$ are both elementary then the result follows from Lemma~\ref{lem:attrdisj}. If $B_1$ is elementary and is the parent of $B_2$ then the result follows from Lemma~\ref{lem:attrjoint}.

  For the inductive case suppose the result holds for some $i$ where $2\leq i<k$. Now both $(\bigcup_{j\in I}B_j)$, where $I$ is the set of indices of the basic blocks in $\ac(B_{i+1})^-$, and $(\bigcup_{j\leq i}B_j)$ are elementary. Now, if $B_{i+1}$ is elementary then the result follows from Lemma~\ref{lem:attrdisj}. If $B_{i+1}$ is non-elementary then $(\bigcup_{j\in I}B_j)$ is the parent of $B_{i+1}$ and the result follows from Lemma~\ref{lem:attrjoint}.
\end{proof}

Next, let us come back to the case where $\bn$ has two blocks $B_1$ and $B_2$.

\begin{lemma}\label{lem:subset}
  Suppose $B_1\cap B_2=\emptyset$ and both $B_1$ and $B_2$ are elementary blocks of $\bn$. Let $A, A_1$ and $A_2$ be attractors of $\ts, \ts_1$ and $\ts_2$ respectively where $A=A_1\cross A_2$. Then $\bas(A)=\bas(A_1)\cross\bas(A_2)$.
\end{lemma}

\begin{proof}
  Follows easily from Lemma~\ref{lem:disjoint}.
\end{proof}

\begin{lemma}\label{lem:subset2}
  Let $A, A_1$ and $A_2$ be the attractors of $\ts, \ts_1$ and $\ts_2$ respectively where $B_1$ and $B_1$ are elementary and non-elementary blocks respectively of $\bn$ with $B_1$ being the parent of $B_2$ and $\ts_2$ being realized by $\bas(A_1)$ and $A=A_2$. Then $\bas(A_1)\cross\bas(A_2) = \bas(A_2) = \bas(A)$.
\end{lemma}

\begin{proof}
  Since $\ts_2$ is realized by $\bas(A_1)$, by its construction (Definition~\ref{def:tsblocks}) we have, for every state $\state\in \ts_2$, $\state_1 \in \bas(A_1)$. Hence $\bas(A_1)\cross\bas(A_2) = \bas(A_2)$.

We next show that $\bas(A_2) = \bas(A)$. Suppose $\state\in\bas(A_2)$. To show that $\state\in \bas(A)$, it is enough to show that:\\
\noindent (i) There is a path from $\state$ to some $\state^A\in A$ in $\ts$ and\\
\noindent (ii) There is no path from $\state$ to $\tstate\in A'$ for some attractor $A'\neq A$ of $\ts$.

(i) Since $\state\in \bas(A_2)$, and $A_2=A$, there is a path $\path$ from $\state$ to $\state^A\in A$ in $\ts_2$. It is easy to see from the construction of $\ts_2$ (Definition~\ref{def:tsblocks}) that $\path$ is also a path in $\ts$ from $\state$ to $\state^A$.

(ii) Suppose for contradiction that there is a path $\path'$ in $\ts$ from $\state$ to $\tstate\in A'$ for some attractor $A'\neq A$ of $\ts$. Since $A'\neq A$ we must have that either (a) $A_1\neq A'_1$ or (b) $A_1=A'_1$ but $A_2\neq A'_2$.

(a) In this case, by Lemma~\ref{lem:fulltoelem}, there must be a path from $\state_1$ to $\tstate_1\in A'_1$ which is a contradiction to the fact that $\state_1\in \bas(A_1)$.

(b) We have by Theorem~\ref{thm:attpres} that $A'_2=A'$. Once again from the construction of $\ts_2$ (Definition~\ref{def:tsblocks}) it is easy to see that $\path'$ is also a path in $\ts_2$ from $\state$ to $\tstate\in A'$. But this contradicts the fact that $\state\in\bas(A_2)$.

For the converse direction suppose that $\state\in\bas(A)$. To show that $\state\in \bas(A_2)$, it is enough to show that:\\
\noindent(iii) There is a path from $\state$ to some $\state^{A_2}\in A_2$ and\\
\noindent(iv) There is no path from $\state$ to $\tstate\in A'_2$ for some attractor $A'_2\neq A_2$ of $\ts_2$.

(iii) Since $\state\in \bas(A)$, there is a path $\path$ in $\ts$ from $\state$ to some $\state^A\in A$. By the fact that $A_2=A$ and by the construction of $\ts_2$ (Definition~\ref{def:tsblocks}) it is clear that $\path$ is also a path in $\ts_2$ from $\state$ to $\state^A\in A_2$.

(iv) Suppose for contradiction that there is a path $\path'$ in $\ts_2$ from $\state$ to $\tstate\in A'_2$ for some attractor $A'_2\neq A_2$ of $\ts_2$. By Theorem~\ref{thm:attpres}, $A'_2$ is equal to an attractor $A'$ of $\ts$ and $A'\neq A$. It is then easy to see again from the construction of $\ts_2$ (Definition~\ref{def:tsblocks}) that $\path'$ is also a path in $\ts$ from $\state$ to $\tstate\in A'$. But this contradicts the assumption that $\state\in \bas(A)$.
\end{proof}

Let us, for the final time, come back to the case where $\bn$ has $k>2$ blocks and these blocks are topologically sorted as $\{B_1,B_2,\ldots, B_k\}$. Let $i$ range over $\{1,2,\ldots, k\}$. By the theorem on attractor preservation, Theorem~\ref{thm:attpres}, we have that ($\cross_{j\leq i} A_j$) is an attractor of $\tshat_i$.

\begin{lemma}\label{lem:bassubset2} Suppose $\bn$ has $k$ basic blocks that are topologically sorted as $\{B_1,\\B_2,\ldots, B_k\}$. Suppose for every attractor $A$ of $\ts$ and for every $i:1\leq i<k$, if $B_{i+1}$ is non-elementary then $\ts_{i+1}$ is realised by $\bas(\cross_{j\in I}A_j)$, its basin w.r.t. the TS for $(\bigcup_{j\in I}B_j)$, where $I$ is the set of indices of the basic blocks in $\ac(B_{i+1})^-$ [where $(\cross_{j\in I}A_j)$, by Theorem~\ref{thm:attpres}, is an attractor of the TS for $(\bigcup_{j\in I}B_j)$]. Then for every $i$, $(\cross_{j\leq i}\bas(A_j)) = \bas(\cross_{j\leq i} A_i)$ where $\bas(\cross_{j\leq i} A_j)$ is the basin of attraction of $(\cross_{j\leq i} A_j)$ with respect to transition system $\tshat_i$ of $(\bigcup_{j\leq i}\! B_j)$.
\end{lemma}

\begin{proof}
  The proof is by induction on $i$. The base case is when $i=2$. Then either $B_1$ and $B_2$ are both elementary and disjoint in which case the proof follows from Lemma~\ref{lem:subset}. Or, $B_1$ is elementary and $B_2$ is non-elementary and $B_1$ is the parent block of $B_2$. In this case the proof follows from Lemma~\ref{lem:subset2}.

  For the inductive case, suppose that the conclusion of the theorem holds for some $i: 2\leq i<k$. Now, consider $(\cross_{j\leq (i+1)}\bas(A_j))$. By the induction hypothesis, we have that $(\cross_{j\leq i}\bas(A_j)) = \bas(\cross_{j\leq i}A_j)$ where $(\cross_{j\leq i}A_j)$ is an attractor of the transition system $\tshat_i$ of the elementary block $(\bigcup_{j\leq i}B_j)$ and $\bas(\cross_{j\leq i}A_j)$ is its basin. Now, either $B_{i+1}$ is elementary in which case we use Lemma~\ref{lem:subset} or $B_{i+1}$ is non-elementary and $(\bigcup_{j\in I}B_j)$ is its parent in which case we use Lemma~\ref{lem:subset2}.

  In either case, we have $(\cross_{j\leq (i+1)}\bas(A_j))=\bas(\cross_{j\leq (i+1)} A_j)$, where $\bas(\cross\\_{j\leq (i+1)}A_j)$ is the basin of attraction of the attractor $(\cross_{j\leq (i+1)}A_j)$ of $\tshat_{i+1}$.
\end{proof}

{
\renewcommand{\thetheorem}{\ref{thm:bassubset}}
\begin{theorem}[preservation of basins] Given the hypothesis and the notations of Lemma~\ref{lem:bassubset2}, we have $(\cross_{i\leq k}\bas(A_i))= \bas(A)$ where $\bas(A)$ is the basin of attraction of the attractor $A=(A_1\cross A_2\cross \ldots\cross A_k)$ of $\ts$.
\end{theorem}
\addtocounter{theorem}{-1}
}

\begin{proof}
  Follows directly by setting $i=k$ in Lemma~\ref{lem:bassubset2}.
\end{proof}

\section{Two Biological Case Studies}
\label{app.casestudy}
\begin{figure*}[!h]
\centering
\includegraphics[width=0.8\textwidth]{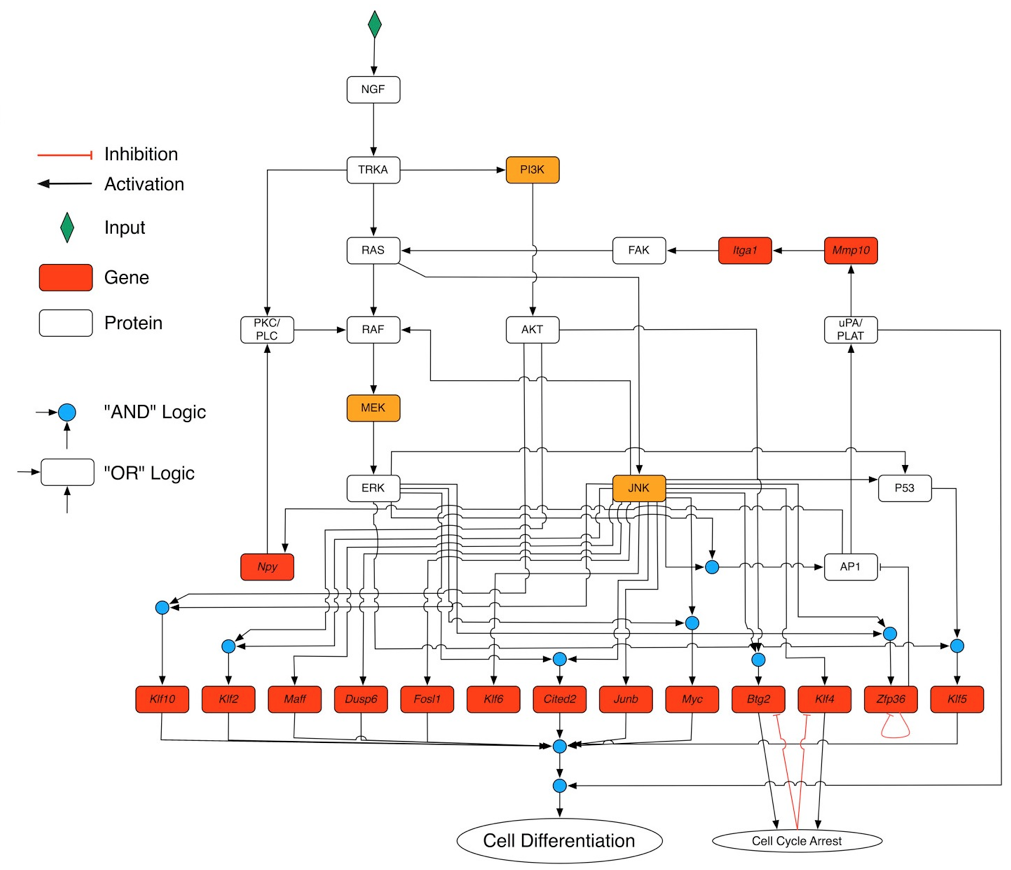}
\caption{Structure of the PC12 cell network of~\cite{OKB16}.}
\label{fig:pc12_structure}
\end{figure*}

\begin{table*}
\begin{tabular}{|C{0.5cm}|C{2.5cm}|C{0.5cm}|C{2cm}|C{0.5cm}|C{2cm}|C{0.5cm}|C{3.5cm}|}
\hline
\multicolumn{1}{|c|}{scc \#} & \multicolumn{1}{c|}{nodes}   & \multicolumn{1}{c|}{scc \#}   & \multicolumn{1}{c|}{nodes}  & \multicolumn{1}{c|}{scc \#} & \multicolumn{1}{c|}{nodes}  & \multicolumn{1}{c|}{scc \#} & \multicolumn{1}{c|}{nodes}              \\ \hline
\multirow{1}{*}{0}& \multirow{1}{*}{NGF} 
&\multirow{1}{*}{5} &\multirow{1}{*}{MAFF} 
&\multirow{1}{*}{10} &\multirow{1}{*}{CITED2} 
&\multirow{1}{*}{15} &BTG2,KLF4,CellCycleArrest\\ \hline
1&TRKA  &6 &KLF6 &11 &PI3K  &16 &AKT \\ \hline  
2&ZFP36,ZFP36\_inh  &7 &JUNB &12 &KLF2 &17 &MYC \\ \hline
3&P53  &8 &FOSL1 &13 &KLF10 &18 &CellDifferentiation \\ \hline
\multirow{2}{*}{4}&\multirow{2}{*}{KLF5}  &\multirow{2}{*}{9} &\multirow{2}{*}{DUSP6} &\multirow{2}{*}{14} & \multicolumn{3}{c|}{\multirow{1}{*}{AP1,ERK,FAK,ITGA1,JNK,MEK,} } \\
 & & & & &   \multicolumn{3}{c|}{\multirow{1}{*}{ NPY,MMP10,PLC,RAF,RAS,UPAR } }   \\ \hline
\end{tabular}
\caption{Nodes in SCCs of the PC12 cell network.}
\label{tab:sccnodes_pc12}
\end{table*}

\begin{table*}
\begin{tabular}{|C{1.3cm}|C{0.05cm}C{0.05cm}C{0.05cm}C{0.05cm}C{0.05cm}C{0.05cm}C{0.05cm}C{0.05cm}C{0.05cm}C{0.05cm}C{0.05cm}C{0.05cm}C{0.05cm}C{0.05cm}C{0.05cm}C{0.05cm}C{0.05cm}C{0.05cm}C{0.05cm}C{0.05cm}C{0.05cm}C{0.05cm}C{0.05cm}C{0.05cm}C{0.05cm}C{0.05cm}C{0.05cm}C{0.05cm}C{0.05cm}C{0.05cm}C{0.05cm}C{0.05cm}C{0.05cm}|}
\hline
Attractor & \multicolumn{33}{c|}{attractor states} \\ \hline
$1$ & $0$ & $0$ & $1$ & $1$ & $1$ & $1$ & $1$ & $0$ & $0$ & $1$ & $1$ & $1$ & $1$ & $1$ & $0$ & $0$ & $1$ & $1$ & $0$ & $1$ & $1$ & $1$ & $1$ & $0$ & $0$ & $1$ & $0$ & $1$ & $1$ & $1$ & $1$ & $1$ & $1$   \\
$2$ & $0$ & $0$ & $1$ & $1$ & $1$ & $1$ & $1$ & $0$ & $0$ & $1$ & $1$ & $1$ & $1$ & $1$ & $0$ & $1$ & $1$ & $1$ & $0$ & $1$ & $1$ & $1$ & $1$ & $0$ & $0$ & $1$ & $0$ & $1$ & $1$ & $1$ & $1$ & $1$ & $0$  \\ 
$3$ & $0$ & $0$ & $0$ & $0$ & $0$ & $0$ & $0$ & $0$ & $0$ & $0$ & $0$ & $0$ & $0$ & $0$ & $0$ & $0$ & $0$ & $0$ & $0$ & $0$ & $0$ & $0$ & $0$ & $0$ & $0$ & $1$ & $0$ & $0$ & $0$ & $0$ & $0$ & $0$ & $1$   \\
$4$ & $0$ & $0$ & $0$ & $0$ & $0$ & $0$ & $0$ & $0$ & $0$ & $0$ & $0$ & $0$ & $0$ & $0$ & $0$ & $0$ & $0$ & $0$ & $0$ & $0$ & $0$ & $0$ & $0$ & $0$ & $0$ & $1$ & $0$ & $0$ & $0$ & $0$ & $0$ & $0$ & $0$   \\
$5$ & $0$ & $0$ & $0$ & $0$ & $0$ & $0$ & $0$ & $0$ & $0$ & $0$ & $0$ & $0$ & $0$ & $0$ & $0$ & $0$ & $0$ & $0$ & $0$ & $0$ & $0$ & $0$ & $0$ & $0$ & $0$ & $0$ & $0$ & $0$ & $0$ & $0$ & $0$ & $0$ & $1$   \\
$6$ & $0$ & $0$ & $0$ & $0$ & $0$ & $0$ & $0$ & $0$ & $0$ & $0$ & $0$ & $0$ & $0$ & $0$ & $0$ & $0$ & $0$ & $0$ & $0$ & $0$ & $0$ & $0$ & $0$ & $0$ & $0$ & $0$ & $0$ & $0$ & $0$ & $0$ & $0$ & $0$ & $0$   \\
$7$ & $1$ & $1$ & $1$ & $1$ & $1$ & $1$ & $1$ & $1$ & $1$ & $1$ & $1$ & $1$ & $1$ & $1$ & $1$ & $0$ & $1$ & $1$ & $1$ & $1$ & $1$ & $1$ & $1$ & $0$ & $0$ & $1$ & $1$ & $1$ & $1$ & $1$ & $1$ & $1$ & $1$   \\
\hline
\end{tabular}
\caption{Attractor states of PC12 cell network. The sequence of the nodes in each state is NGF, TRKA, RAS, RAF, JNK, MEK, ERK, PI3K, AKT, PLC, NPY, JUNB, P53, AP1, KLF2, KLF4, KLF5, KLF6, KLF10, MAFF, DUSP6, FOSL1, CITED2, BTG2, ZFP36, ZFP36\_inh, MYC, UPAR, MMP10, ITGA1, FAK, CellDifferentiation, CellCycleArrest. }
\label{tab:pc12att}
\end{table*}

\begin{table*}[!t] 
\centering 
\begin{tabular}{|C{1.5cm}|C{1.5cm}C{1.5cm}C{1.5cm}C{1.5cm}C{1.5cm}C{1.5cm}C{1.5cm}|}
 \hline 
\multicolumn{1}{|c|}{\multirow{2}{*}{Attractor}} & \multicolumn{7}{c|}{Time (ms)}
 \\ \cline{2-8} 
\multicolumn{1}{|c|}{} & \multicolumn{1}{c}{~~~~~~$1$~~~~~~} & \multicolumn{1}{c}{~~~~~~$2$~~~~~~} & \multicolumn{1}{c}{~~~~~~$3$~~~~~~} & \multicolumn{1}{c}{~~~~~~$4$~~~~~~} & \multicolumn{1}{c}{~~~~~~$5$~~~~~~} & \multicolumn{1}{c}{~~~~~~$6$~~~~~~} & \multicolumn{1}{c|}{~~~~~~$7$~~~~~~} 
 \\ \hline
\multirow{2}{*}{$1$} &$-$ &$24$ &$55$ &$34$ &$31$ &$25$ &$16$ \\ 
&$-$ &$5$ &$6$ &$6$ &$6$ &$6$ &$12$ \\ \hline 
\multirow{2}{*}{$2$} &$26$ &$-$ &$54$ &$35$ &$31$ &$25$ &$16$ \\ 
&$5$ &$-$ &$6$ &$6$ &$6$ &$6$ &$12$ \\ \hline 
\multirow{2}{*}{$3$} &$26$ &$20$ &$-$ &$46$ &$31$ &$17$ &$23$ \\ 
&$5$ &$5$ &$-$ &$6$ &$5$ &$6$ &$12$ \\ \hline 
\multirow{2}{*}{$4$} &$25$ &$20$ &$56$ &$-$ &$30$ &$19$ &$27$ \\ 
&$5$ &$5$ &$6$ &$-$ &$6$ &$5$ &$12$ \\ \hline 
\multirow{2}{*}{$5$} &$25$ &$19$ &$52$ &$34$ &$-$ &$20$ &$22$ \\ 
&$5$ &$5$ &$6$ &$6$ &$-$ &$6$ &$12$ \\ \hline 
\multirow{2}{*}{$6$} &$25$ &$20$ &$56$ &$33$ &$32$ &$-$ &$23$ \\ 
&$5$ &$5$ &$6$ &$6$ &$6$ &$-$ &$12$ \\ \hline 
\multirow{2}{*}{$7$} &$25$ &$20$ &$56$ &$33$ &$32$ &$24$ &$-$ \\ 
&$5$ &$5$ &$6$ &$6$ &$6$ &$6$ &$-$ \\ \hline 
 \end{tabular}
\caption{Time costs of the global approach and the decomposition-based approach on the PC12 cell network.}
\label{}
\end{table*}

\begin{figure*}[!t]
\centering
\includegraphics[width=.8\textwidth,angle=0]{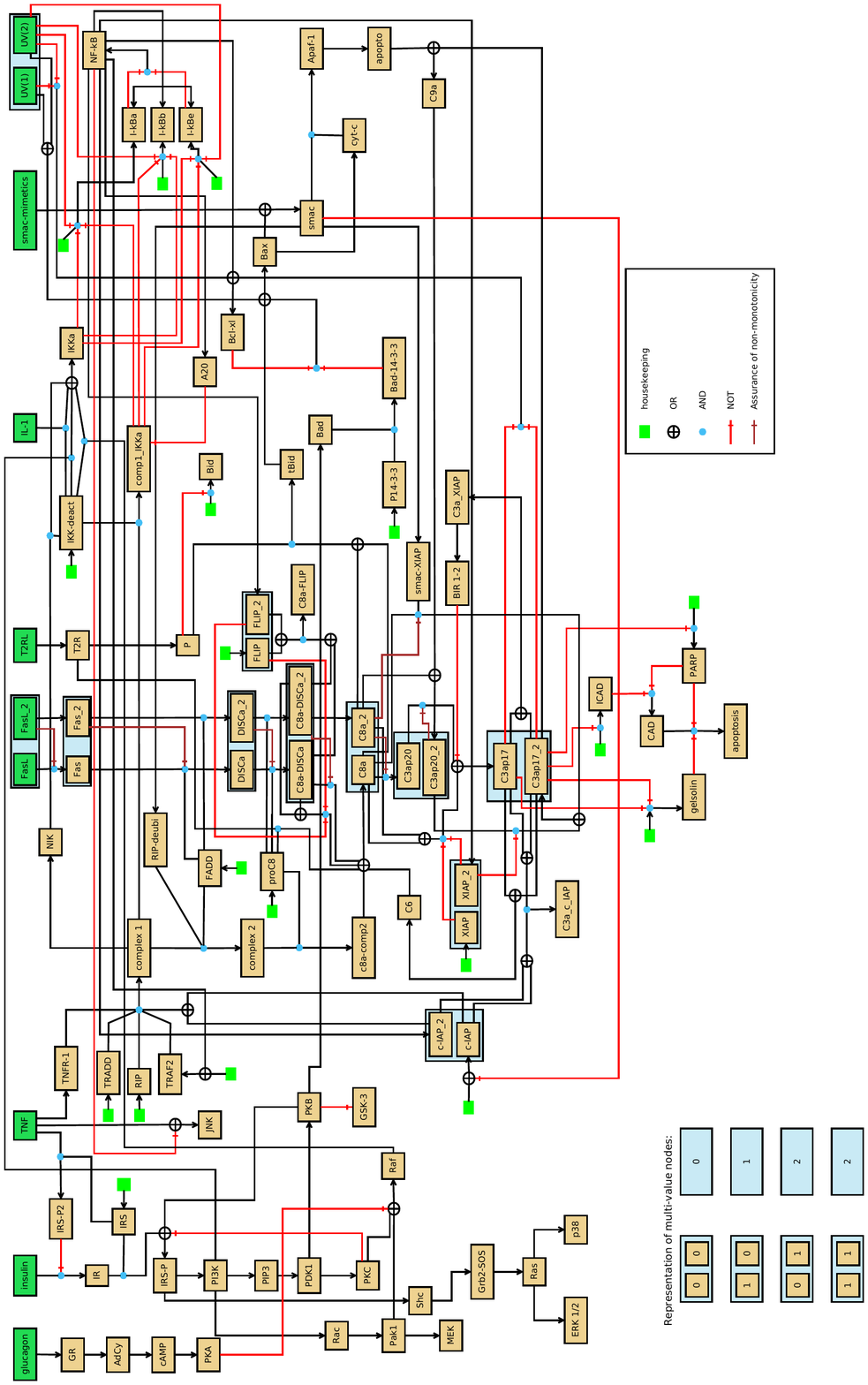}
\caption{The wiring of the multi-value logic model of apoptosis by Schlatter et
al.~\cite{SSVSSBEMS09} recast into a~binary Boolean network. For clarity of the diagram the nodes
I-kBa, I-kBb, and I-kBe have two positive inputs. The inputs are interpreted as connected via
$\oplus$ (logical OR).}
\label{fig:structure}
\end{figure*}

\begin{table*}[!t]
\centering
\begin{tabular}{|C{0.5cm}|C{2.4cm}|C{0.5cm}|C{2.4cm}|C{0.5cm}|C{2.4cm}|C{0.5cm}|C{2.3cm}|}
\hline
\multicolumn{1}{|c|}{scc \#} & \multicolumn{1}{c|}{nodes} & \multicolumn{1}{c|}{scc \#} & \multicolumn{1}{c|}{nodes} & \multicolumn{1}{c|}{scc \#} & \multicolumn{1}{c|}{nodes} & \multicolumn{1}{c|}{scc \#} & \multicolumn{1}{c|}{nodes} \\ \hline
0                            & apoptosis                  & 17                          & C8a\_DISCa\_2              & 32                          & IRS\_P2                    & 47                          & UV                         \\ \hline
1                            & gelsolin                   & 18                          & C8a\_DISCa                 & 33                          & IRS                        & 48                          & UV\_2                      \\ \hline
2                            & C3a\_c\_IAP                & 19                          & proC8                      & 34                          & IKKdeact                   & 49                          & FASL                       \\ \hline
3                            & I\_kBb                     & 20                          & p38                        & 35                          & FLIP                       & 50                          & PKA                        \\ \hline
4                            & CAD                        & 21                          & ERK1o2                     & 36                          & DISCa\_2                   & 51                          & cAMP                       \\ \hline
5                            & PARP                       & 22                          & Ras                        & 37                          & DISCa                      & 52                          & AdCy                       \\ \hline
6                            & ICAD                       & 23                          & Grb2\_SOS                  & 38                          & FADD                       & 53                          & GR                         \\ \hline
7                            & JNK                        & 24                          & Shc                        & 39                          & Bid                        & 54                          & Glucagon                   \\ \hline
8                            & C8a\_FLIP                  & 25                          & Raf                        & 40                          & housekeeping               & 55                          & Insulin                    \\ \hline
9                           & XIAP                       & 26                          & MEK                        & 41                          & FAS\_2                     & 56                          & smac\_mimetics             \\ \hline
10                           & TRADD                      & 27                          & Pak1                       & 42                          & FAS                        & 57                          & P                          \\ \hline
11                           & RIP                        & 28                          & Rac                        & 43                          & FASL\_2                    & 58                          & T2R                        \\ \hline
12                           & Bad\_14\_3\_3              & 29                          & GSK\_3                     & 44                          & IL\_1                      & 59                          & T2RL                       \\ \hline
13                           & P14\_3\_3                  & 30                          & Bad                        & 45                          & TNFR\_1                    &                             &                            \\ \hline
14                           & C8a\_2                     & 31                          & IR                         & 46                          & TNF                        &                             &                            \\ \hline
\multicolumn{1}{|c|}{\multirow{4}{*}{15}} & \multicolumn{7}{l|}{\multirow{4}{*}{\begin{tabular}[c]{@{}l@{}}Apaf\_1 apopto A20 Bax Bcl\_xl BIR1\_2 c\_IAP c\_IAP\_2 complex1  comp1\_IKKa cyt\_c C3ap20\\ C3ap20\_2 C3a\_XIAP  C8a\_comp2 C9a  FLIP\_2 NIK RIP\_deubi smac smac\_XIAP tBid TRAF2\\ XIAP\_2 IKKa I\_kBa I\_kBe complex2 NF\_kB C8a C3ap17 C3ap17\_2\end{tabular}}} \\
\multicolumn{1}{|r|}{}                   & \multicolumn{7}{l|}{}                                                                                                                                                                                                                                                                                                                 \\
\multicolumn{1}{|r|}{}                   & \multicolumn{7}{l|}{}                                                                                                                                                                                                                                                                                                                 \\
\multicolumn{1}{|r|}{}                   & \multicolumn{7}{l|}{}                                                                                                                                                                                                                                                                                                                 \\ \hline
\multicolumn{1}{|c|}{16}                 & \multicolumn{7}{l|}{IRS\_P PDK1 PKB PKC PIP3 PI3K C6}                                                                                                                                                                                                                                                                                 \\ \hline
\end{tabular}
\caption{Nodes in SCCs of the apoptosis network}
\label{tab:sccnodes_apoptosis}
\end{table*}

\begin{table*}[!t] 
\centering 
\begin{tabular}{|c|cccccccccccccccc|}
 \hline 
\multicolumn{1}{|c|}{\multirow{2}{*}{Attractor}} & \multicolumn{16}{c|}{Time (s)}
 \\ \cline{2-17} 
\multicolumn{1}{|c|}{} & \multicolumn{1}{c}{~~~~~~$1$~~~~~~} & \multicolumn{1}{c}{~~~~~~$2$~~~~~~} & \multicolumn{1}{c}{~~~~~~$3$~~~~~~} & \multicolumn{1}{c}{~~~~~~$4$~~~~~~} & \multicolumn{1}{c}{~~~~~~$5$~~~~~~} & \multicolumn{1}{c}{~~~~~~$6$~~~~~~} & \multicolumn{1}{c}{~~~~~~$7$~~~~~~} & \multicolumn{1}{c}{~~~~~~$8$~~~~~~} & \multicolumn{1}{c}{~~~~~~$9$~~~~~~} & \multicolumn{1}{c}{~~~~~~$10$~~~~~~} & \multicolumn{1}{c}{~~~~~~$11$~~~~~~} & \multicolumn{1}{c}{~~~~~~$12$~~~~~~} & \multicolumn{1}{c}{~~~~~~$13$~~~~~~} & \multicolumn{1}{c}{~~~~~~$14$~~~~~~} & \multicolumn{1}{c}{~~~~~~$15$~~~~~~} & \multicolumn{1}{c|}{~~~~~~$16$~~~~~~} 
 \\ \hline
\multirow{2}{*}{$1$} &$-$ &$5.45$ &$1.69$ &$1.80$ &$23.12$ &$18.20$ &$24.53$ &$24.48$ &$1.54$ &$1.47$ &$1.72$ &$1.62$ &$45.67$ &$34.26$ &$45.08$ &$44.00$ \\ 
&$-$ &$0.79$ &$0.81$ &$0.79$ &$0.99$ &$0.99$ &$0.99$ &$0.99$ &$0.77$ &$0.75$ &$0.77$ &$0.75$ &$0.94$ &$0.91$ &$0.94$ &$0.91$ \\ \hline 
\multirow{2}{*}{$2$} &$1.67$ &$-$ &$1.71$ &$1.87$ &$17.98$ &$18.45$ &$25.13$ &$23.89$ &$1.58$ &$1.48$ &$1.76$ &$1.67$ &$34.53$ &$34.40$ &$45.31$ &$43.98$ \\ 
&$0.81$ &$-$ &$0.78$ &$0.81$ &$0.94$ &$0.99$ &$0.99$ &$0.99$ &$0.78$ &$0.75$ &$0.77$ &$0.75$ &$0.95$ &$0.91$ &$0.94$ &$0.90$ \\ \hline 
\multirow{2}{*}{$3$} &$1.65$ &$1.55$ &$-$ &$1.85$ &$17.23$ &$18.21$ &$24.97$ &$24.22$ &$1.58$ &$1.49$ &$1.74$ &$1.66$ &$34.40$ &$34.17$ &$45.16$ &$44.22$ \\ 
&$0.81$ &$0.79$ &$-$ &$0.81$ &$0.94$ &$0.99$ &$0.99$ &$0.99$ &$0.77$ &$0.75$ &$0.77$ &$0.75$ &$0.95$ &$0.91$ &$0.94$ &$0.90$ \\ \hline 
\multirow{2}{*}{$4$} &$1.66$ &$1.57$ &$1.86$ &$-$ &$17.46$ &$19.09$ &$24.22$ &$24.55$ &$1.59$ &$1.49$ &$1.76$ &$1.65$ &$34.71$ &$34.48$ &$45.45$ &$44.42$ \\ 
&$0.81$ &$0.79$ &$0.81$ &$-$ &$0.94$ &$0.99$ &$0.99$ &$0.99$ &$0.78$ &$0.75$ &$0.77$ &$0.75$ &$0.95$ &$0.91$ &$0.94$ &$0.91$ \\ \hline 
\multirow{2}{*}{$5$} &$1.69$ &$1.59$ &$1.89$ &$1.78$ &$-$ &$20.94$ &$24.11$ &$24.70$ &$1.60$ &$1.51$ &$1.75$ &$1.67$ &$34.26$ &$34.34$ &$45.32$ &$44.28$ \\ 
&$0.81$ &$0.78$ &$0.81$ &$0.79$ &$-$ &$0.99$ &$0.99$ &$0.99$ &$0.78$ &$0.75$ &$0.77$ &$0.75$ &$0.95$ &$0.91$ &$0.95$ &$0.91$ \\ \hline 
\multirow{2}{*}{$6$} &$1.68$ &$1.58$ &$1.86$ &$1.74$ &$20.87$ &$-$ &$24.57$ &$24.44$ &$1.60$ &$1.50$ &$1.78$ &$1.67$ &$34.81$ &$34.68$ &$45.44$ &$43.95$ \\ 
&$0.82$ &$0.79$ &$0.81$ &$0.79$ &$0.99$ &$-$ &$0.99$ &$0.99$ &$0.77$ &$0.75$ &$0.77$ &$0.75$ &$0.94$ &$0.91$ &$0.95$ &$0.91$ \\ \hline 
\multirow{2}{*}{$7$} &$1.68$ &$1.58$ &$1.86$ &$1.74$ &$20.96$ &$17.30$ &$-$ &$24.31$ &$1.59$ &$1.49$ &$1.77$ &$1.66$ &$34.71$ &$34.69$ &$45.43$ &$43.92$ \\ 
&$0.81$ &$0.79$ &$0.82$ &$0.80$ &$0.99$ &$0.99$ &$-$ &$0.99$ &$0.77$ &$0.75$ &$0.77$ &$0.75$ &$0.95$ &$0.91$ &$0.94$ &$0.90$ \\ \hline 
\multirow{2}{*}{$8$} &$1.67$ &$1.57$ &$1.85$ &$1.74$ &$21.02$ &$17.25$ &$24.69$ &$-$ &$1.59$ &$1.50$ &$1.77$ &$1.67$ &$34.60$ &$34.66$ &$45.41$ &$43.96$ \\ 
&$0.81$ &$0.79$ &$0.81$ &$0.78$ &$0.99$ &$0.99$ &$0.99$ &$-$ &$0.77$ &$0.75$ &$0.77$ &$0.75$ &$0.95$ &$0.91$ &$0.94$ &$0.90$ \\ \hline 
\multirow{2}{*}{$9$} &$1.68$ &$1.57$ &$1.83$ &$1.74$ &$20.05$ &$17.25$ &$24.65$ &$24.31$ &$-$ &$1.52$ &$1.66$ &$1.80$ &$34.35$ &$34.33$ &$44.97$ &$44.37$ \\ 
&$0.81$ &$0.78$ &$0.81$ &$0.78$ &$0.99$ &$0.99$ &$0.99$ &$0.99$ &$-$ &$0.77$ &$0.75$ &$0.77$ &$0.91$ &$0.95$ &$0.90$ &$0.94$ \\ \hline 
\multirow{2}{*}{$10$} &$1.67$ &$1.57$ &$1.85$ &$1.74$ &$20.77$ &$17.38$ &$24.69$ &$24.57$ &$1.59$ &$-$ &$1.64$ &$1.78$ &$34.51$ &$33.74$ &$45.64$ &$44.00$ \\ 
&$0.78$ &$0.81$ &$0.78$ &$0.81$ &$0.94$ &$0.99$ &$0.99$ &$0.99$ &$0.77$ &$-$ &$0.75$ &$0.77$ &$0.91$ &$0.95$ &$0.90$ &$0.94$ \\ \hline 
\multirow{2}{*}{$11$} &$1.67$ &$1.57$ &$1.85$ &$1.74$ &$20.88$ &$17.03$ &$24.41$ &$24.18$ &$1.60$ &$1.50$ &$-$ &$1.78$ &$33.37$ &$34.14$ &$46.48$ &$44.89$ \\ 
&$0.78$ &$0.81$ &$0.78$ &$0.81$ &$0.94$ &$0.99$ &$0.99$ &$0.99$ &$0.77$ &$0.75$ &$-$ &$0.77$ &$0.91$ &$0.94$ &$0.90$ &$0.94$ \\ \hline 
\multirow{2}{*}{$12$} &$1.71$ &$1.60$ &$1.86$ &$1.77$ &$20.76$ &$16.89$ &$24.98$ &$24.61$ &$1.61$ &$1.52$ &$1.79$ &$-$ &$34.71$ &$34.01$ &$45.86$ &$44.58$ \\ 
&$0.79$ &$0.81$ &$0.78$ &$0.81$ &$0.94$ &$0.99$ &$0.99$ &$0.99$ &$0.78$ &$0.75$ &$0.77$ &$-$ &$0.91$ &$0.94$ &$0.91$ &$0.94$ \\ \hline 
\multirow{2}{*}{$13$} &$1.70$ &$1.59$ &$1.88$ &$1.77$ &$21.30$ &$17.43$ &$24.72$ &$24.48$ &$1.62$ &$1.53$ &$1.80$ &$1.69$ &$-$ &$34.69$ &$46.56$ &$44.45$ \\ 
&$0.79$ &$0.81$ &$0.78$ &$0.81$ &$0.94$ &$0.99$ &$0.99$ &$0.99$ &$0.78$ &$0.75$ &$0.77$ &$0.75$ &$-$ &$0.94$ &$0.90$ &$0.94$ \\ \hline 
\multirow{2}{*}{$14$} &$1.69$ &$1.59$ &$1.88$ &$1.78$ &$21.01$ &$17.26$ &$24.79$ &$24.27$ &$1.62$ &$1.53$ &$1.79$ &$1.68$ &$35.28$ &$-$ &$46.11$ &$44.69$ \\ 
&$0.79$ &$0.81$ &$0.78$ &$0.82$ &$0.94$ &$0.99$ &$0.99$ &$0.99$ &$0.78$ &$0.75$ &$0.77$ &$0.75$ &$0.94$ &$-$ &$0.90$ &$0.94$ \\ \hline 
\multirow{2}{*}{$15$} &$1.68$ &$1.59$ &$1.87$ &$1.76$ &$20.96$ &$17.91$ &$24.57$ &$24.39$ &$1.62$ &$1.52$ &$1.80$ &$1.71$ &$35.15$ &$34.37$ &$-$ &$44.72$ \\ 
&$0.79$ &$0.81$ &$0.79$ &$0.81$ &$0.94$ &$0.99$ &$0.99$ &$0.98$ &$0.77$ &$0.75$ &$0.77$ &$0.75$ &$0.94$ &$0.91$ &$-$ &$0.94$ \\ \hline 
\multirow{2}{*}{$16$} &$1.70$ &$1.60$ &$1.89$ &$1.78$ &$21.29$ &$17.35$ &$24.97$ &$24.33$ &$1.62$ &$1.53$ &$1.79$ &$1.68$ &$34.97$ &$34.38$ &$45.58$ &$-$ \\ 
&$0.79$ &$0.81$ &$0.78$ &$0.81$ &$0.94$ &$0.99$ &$0.99$ &$0.99$ &$0.77$ &$0.75$ &$0.77$ &$0.75$ &$0.95$ &$0.91$ &$0.94$ &$-$ \\ \hline 
\end{tabular}
\caption{Time costs of the global approach and the decomposition-based approach on the apoptosis network. In each cell, the upper value is the time cost of the global approach and the lower value is the time cost of the decomposition-based approach.}
\label{}
\end{table*}

\end{document}